\documentclass[sigconf, authorversion]{acmart}

\AtBeginDocument{%
	\providecommand\BibTeX{{%
			\normalfont B\kern-0.5em{\scshape i\kern-0.25em b}\kern-0.8em\TeX}}}

\copyrightyear{2023} 
\acmYear{2023} 
\setcopyright{acmlicensed}
\acmConference[WSDM '23]{Proceedings of the Sixteenth ACM International Conference on Web Search and Data Mining}{February 27-March 3, 2023}{Singapore, Singapore}
\acmBooktitle{Proceedings of the Sixteenth ACM International Conference on Web Search and Data Mining (WSDM '23), February 27-March 3, 2023, Singapore, Singapore}
\acmPrice{15.00}
\acmDOI{10.1145/3539597.3570442}
\acmISBN{978-1-4503-9407-9/23/02}

\usepackage{amsmath,amsfonts,amsthm}
\usepackage{url}
\usepackage{color}
\usepackage[capitalise]{cleveref}
\usepackage{graphicx}
\usepackage{caption} \usepackage{subcaption, enumitem}
\usepackage{hhline}
\usepackage{algorithm}
\usepackage{graphicx}
\usepackage{colortbl}
\usepackage{algpseudocode}
\usepackage{float}
\usepackage{threeparttable}
\usepackage{wrapfig}
\makeatletter

\newcommand{\R}{\mathbb{R}}

\newcommand{\eqdef}{\mathbin{\stackrel{\rm def}{=}}}
\newcommand{\norm}[1]{\|#1\|}

\definecolor{lRed}{cmyk}{0,0.28,0.28,0.0}
\definecolor{lBlue}{cmyk}{0.25,0.17,0.0,0.0}

\usepackage{color}

\makeatletter

\newtheorem*{rep@theorem}{\rep@title}
\newcommand{\newreptheorem}[2]{%
	\newenvironment{rep#1}[1]{%
		\def\rep@title{#2 \ref{##1}}%
		\begin{rep@theorem}}%
		{\end{rep@theorem}}}
\makeatother
\newtheorem{theorem}{Theorem}
\newreptheorem{theorem}{Theorem}

\newtheorem{corollary}[theorem]{Corollary}
\newtheorem{lemma}[theorem]{Lemma}
\newtheorem*{lemma*}{Lemma}
\newreptheorem{lemma}{Lemma}
\newreptheorem{corollary}{Corollary}
\newtheorem{fact}[theorem]{Fact}

\newtheorem{definition}[theorem]{Definition}

  \usepackage{nth}
  \usepackage{intcalc}

\begin{document}
	
\title{Local Edge Dynamics and Opinion Polarization}
	
\author{Nikita Bhalla}
\affiliation{
  \institution{University of Massachusetts Amherst}
  \country{Amherst, MA, USA}
}
\email{nbhalla@cs.umass.edu}

\author{Adam Lechowicz}
\affiliation{
  \institution{University of Massachusetts Amherst}
  \country{Amherst, MA, USA}
}
\email{alechowicz@cs.umass.edu}

\author{Cameron Musco}
\affiliation{
  \institution{University of Massachusetts Amherst}
  \country{Amherst, MA, USA}
}
\email{cmusco@cs.umass.edu}

	
	\begin{abstract}
		The proliferation of social media platforms, recommender systems, and their joint societal impacts have prompted significant interest in opinion formation and evolution within social networks. We study how \emph{local edge dynamics} can drive opinion polarization. In particular, we introduce a variant of the classic Friedkin-Johnsen opinion  dynamics, augmented with a simple time-evolving network model. Edges are iteratively added or deleted according to simple rules, modeling decisions based on individual preferences and network recommendations.
		
		Via simulations on synthetic and real-world graphs, we find that the combined presence of two  dynamics gives rise to high polarization: 1) \emph{confirmation bias} -- i.e., the preference for nodes to connect to other nodes with similar expressed opinions and 2) \emph{friend-of-friend link recommendations}, which encourage new connections between closely connected nodes. 
		We show that our model is tractable to theoretical analysis, which helps  explain how these local dynamics erode connectivity across opinion groups,  affecting polarization and a related measure of disagreement across edges.  
		Finally, we validate our model against real-world data, showing that our edge dynamics drive the structure of arbitrary graphs, including random graphs, to more closely resemble real  social networks.
		
	Our code and supplemental materials are available at \textcolor{blue}{\url{https://github.com/adamlechowicz/opinion-polarization/}}.
	
	\end{abstract}

	
	
	\maketitle
	
	\vspace{-.5em}
	\section{Introduction}
	Over the last twenty years, the rise of massive social media platforms has significantly increased 
	information sharing and human interaction around the globe.  
	While information availability and richer online interactions are positive effects of this social shift, there is increasing concern about the ability of these platforms to polarize and divide us \cite{HartNisbet:2012, LeeChoiKim:2014, DellaPosta:2020, BaldassarriGelman:2008, GuerraMeiraKleinberg:2013}.  
	This polarization seems to be driven by both established factors of human social dynamics, along with new dynamics, driven by the  behavior of the social media platforms themselves.
	
	An example of a human behavior driving polarization is \textit{confirmation bias}, the tendency to avoid information that challenges our own views and to seek information that confirms them \cite{Nickerson:1998}. Confirmation bias is amplified on social media platforms due to the increased availability of opinion-confirming content. Indeed, it is thought to be a key driver of the online spread of conspiracy theories and fake news in recent years \cite{MeppelinkSmitFransenDiviani:2019,Tandoc:2019, PatashnikSchiller:2021}. 
	
	An example of the behavior of social media platforms driving polarization is the use of \emph{recommender systems} to filter and deliver content that maximizes user engagement. 
	Such recommendations can create \emph{filter bubbles}, which further strengthen the power of confirmation bias and drive polarization \cite{Pariser:2012, LeeChoiKim:2014, GuerraMeiraKleinberg:2013}.
	
	\vspace{-.5em}
	\subsection{Our Model}
	\vspace{-.25em}
	We seek to understand how local edge updates (i.e., insertions and deletions) driven by both human behavior and the behavior of social media platforms may cause opinion polarization. To do so we introduce a simple model of network and opinion coevolution, which 1) is based on established opinion dynamics models, 2) captures the effects of both confirmation bias and recommender systems, and 3) remains tractable to theoretical analysis and efficient simulation. This model is detailed in Section \ref{sec:prelim} and summarized below. Beyond our work, we hope that the model will serve as a useful platform for further  investigation of network and opinion coevolution.
	
	\smallskip
	
	\noindent\textbf{Opinion Dynamics and Polarization.} We build on the classic Friedkin-Johnsen opinion  model, \cite{FriedkinJohnsen:1990}, which models how individuals' \emph{expressed opinions} (represented as real numbers) are influenced by their \emph{innate opinions}, along with the expressed opinions of their neighbors in a network. A node's innate opinion is fixed at initialization, and its expressed opinion at any time is an average of this innate opinion with the expressed opinions of its neighbors. 
	
	We consider two metrics studied in prior work \cite{MuscoMuscoTsourakakis:2018}: \emph{opinion polarization}, which is the variance of the expressed opinions, and \emph{disagreement}, which is the total squared difference of \emph{expressed opinions} summed over all edges in the network. Under the Friedkin-Johnsen model, polarization and disagreement tend to counteract each other -- with low disagreement, nodes tend to be connected to other nodes with similar opinions, and polarization tends to be high. With high disagreement, nodes are connected to other nodes with a diversity of opinions, limiting polarization. Both values can be efficiently computed in closed form \cite{MuscoMuscoTsourakakis:2018,XuBaoZhang:2021}, making them tractable for both theoretical and empirical investigation.
	
	
	\smallskip
	
	\noindent\textbf{Edge Dynamics.}
	We consider a model where the network and expressed node opinions coevolve over time.  At each time step, users within the network stochastically delete connections to some of their neighbors and add new connections. To model confirmation bias,
	edges that are more disagreeable are deleted with higher probability. To model recommender systems, new edges are added to random \emph{friend-of-friends} -- i.e., two-hop connections in the network. Friend-of-friend recommendations are a popular form of edge recommendations made by social media platforms \cite{SuSharmaGoel:2016,DalyGeyerMillen:2010}. 
	
	We also fix a small percentage of edges at initialization, 
	which are not considered for deletion at any time step.
	These edges can be thought of as modeling connections that are independent of opinions or recommendations, e.g., to family members or co-workers.  
	
	\smallskip
	
	\noindent\textbf{Simulation and Theoretical Analysis.}
	To study how the above opinion and edge dynamics interact to drive opinion polarization, we employ both simulation and theoretical analysis.
	In our simulations,
	we start with an initial graph, either randomly generated from an established model, or taken from a snapshot of a real social network. We also start with innate opinions, which are randomly generated over the interval $[-1,1]$. 
	We then simulate our edge and opinion dynamics, recording how opinions, polarization, disagreement, and the graph structure evolve over time. 
	Our theoretical results, which help explain many of the dynamics observed in simulation, leverage  well-established formulas for polarization and disagreement in the Friedkin-Johnsen model, defined in Section \ref{sec:prelim}.  
	

	\subsection{Our Findings}
	Our main findings are as follows:
	
	\smallskip
	
	\noindent\textbf{1. Confirmation bias and recommendations drive polarization.} We find that, when both confirmation bias and friend-of-friend link recommendations are part of the edge dynamics model, 
	the network becomes polarized -- nodes sort into distinct clusters of similar and opposing opinions.  When the network has no fixed edges, edge dynamics splinter it into clusters such that expressed opinions are very close to the innate opinions, and polarization nearly reaches its maximum value (see Secs. \ref{sec:happens} and \ref{sec:evolves}).

	
	If either one of these dynamics is removed (i.e., either edge additions or removals are just made randomly), polarization remains low. This finding is generally robust to the initial graph structure and innate opinion distribution, with the exception of dense random graphs, where friend-of-friend recommendations lose their power -- they behave similarly to random recommendations. 
	See Fig. \ref{fig:compIntro} for an illustration and Sec.\ref{sec:density} for more detailed discussion.
	
	While perhaps intuitive, the above finding is far from obvious. In the Friedkin-Johnsen model, a node's expressed opinion is heavily influenced by those of its neighbors. So, initially, a node's expressed opinion will depend very little on its innate opinion. Thus, it is not clear that removing disagreeable edges (where disagreement is with respect to the \emph{expressed} opinions) and adding friend-of-friends will lead to polarization. 
	We also find it surprising that \emph{both} friend-of-friend recommendations and confirmation bias are needed to drive polarization, initially conjecturing that just confirmation bias would suffice. Recommender systems seem to act as a key catalyst in amplifying human behavior that favors polarization.  
	
	\begin{figure}[h]
		\minipage{0.49\textwidth}
		\includegraphics[width=\linewidth]{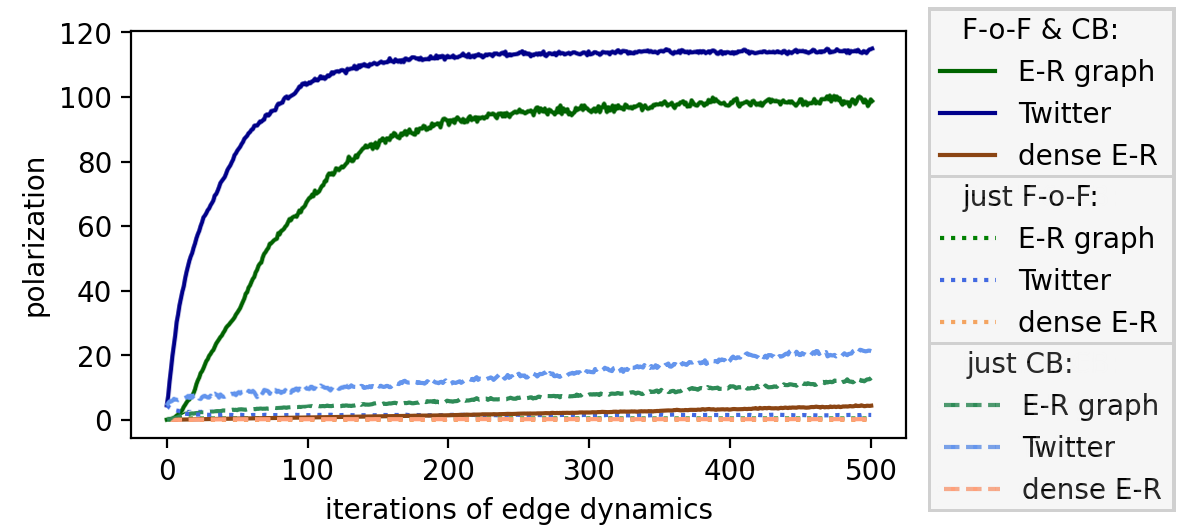}
		\endminipage
		\vspace{-0.5em}
		\caption{Opinion polarization over time for an Erd\"{o}s-Renyi (ER) graph with $1000$ nodes and connection probability $p =~0.05$, a real-world Twitter social network \cite{DeAbir:2014}, and a \textit{dense} ER graph with 1000 nodes and connection probability $p = 0.1$.  Each network has 5\% fixed edges.  Innate opinions are uniformly distributed in $[-1,1]$. Except in the dense ER graph, with edge dynamics influenced by both \textit{confirmation bias} (CB) and \textit{friend-of-friend recommendations} (F-o-F), polarization rises significantly before asymptoting.  When either factor is removed, polarization remains  low.  
		\vspace{-1em}}\label{fig:compIntro}
	\end{figure}
	
	\smallskip
	
	\noindent\textbf{2. Our model is tractable to theoretical analysis.} 
	To complement our simulation results, we  give theoretical bounds which illustrate how  edge removals and insertions can drive polarization and disagreement in the Friedkin-Johnsen model. We prove that swapping an edge with large expressed opinion disagreement across it for a new edge with small disagreement monotonically drives up the \emph{sum of polarization and disagreement} in the graph, a combined quantity studied in prior work \cite{MuscoMuscoTsourakakis:2018}, whose rise seems closely correlated with a rise in polarization itself. 
	
	This finding helps explain how  both confirmation bias and friend-of-friend recommendations drive polarization. 
	In each iteration of our edge dynamics, we remove disagreeable edges, and replace them with friend-of-friend connections. Initially, these friend-of-friend connections are somewhat random (i.e., uncorrelated with the expressed node opinions), but still more agreeable on average than the removed connections. Eventually, as opinion groups separate and the removed edges become less  disagreeable, so do the friend-of-friend connections. Thus, polarization continues rising.
	
	
	We also give an understanding of our model based on the Stochastic Block Model (SBM).  Using a simple method to split nodes into two opinion groups (i.e., one group with innate opinions $< 0$ and one group with innate opinions $> 0$), we coarsely approximate the underlying network as an SBM graph with two blocks, drawn from a distribution with in-group connection probability $p$ and out-group connection probability $q$.  We set $p$ and $q$ to match the in-group and out-group connection densities in the true graph. 
	
	Following prior work \cite{ChitraMusco:2020}, we employ closed form expressions for the polarization and disagreement in the expected SBM graph. 
	We demonstrate empirically that  these expressions yield good approximations to the true polarization and disagreement over time -- see Fig. \ref{fig:fixedIntro}. This indicates  that the evolution of polarization and disagreement in our model can be understood largely in terms of  in-group and out-group connection densities. With confirmation bias and friend-of-friend recommendations, the out-group connection density is eroded over time.  If either dynamic is replaced with the random control, we observe that this connection density does not decrease to the level necessary for high polarization.
	
	Our SBM-based approximations  also predict and explain interesting phenomena in our model. For instance, in networks with no fixed edges, we see that the graph eventually splinters into many connected components with very similar innate opinions, causing polarization to reach a maximum and disagreement to drop to near zero. However, in the presence of a small fraction of fixed edges, while the network becomes more polarized, it remains connected. Surprisingly, in this setting, disagreement, like polarization, tends to increase over time, before asymptoting -- see Fig. \ref{fig:fixedIntro}.
	
	\begin{figure}[h]
		\minipage{0.48\textwidth}
		\includegraphics[width=\linewidth]{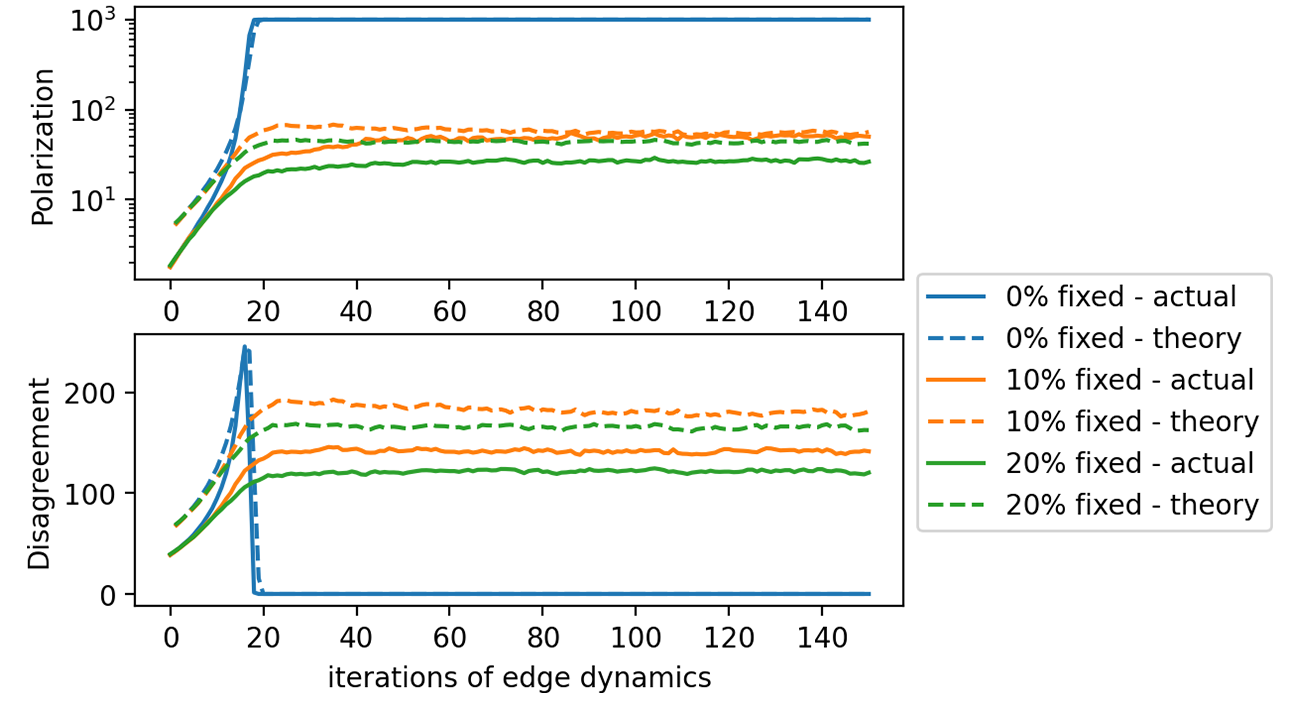}
		\vspace{-1em}
		\endminipage\hfill
		\vspace{-1em}
		\caption{Polarization and disagreement  for an Erd\"{o}s-Renyi graph with $1000$ nodes, varying percentages of fixed edges, and average degree $25$. Innate opinions are uniformly distributed in $[-1,1]$. Estimates of polarization and disagreement based on our SBM approximation closely match observed values and reflect important patterns -- e.g., that even a small percentage of fixed edges significantly limits polarization and prevents disagreement from dropping over time.
			\vspace{-1em}}\label{fig:fixedIntro}
		
	\end{figure}
	
	
	\noindent\textbf{3. Our model creates ``natural-looking'' networks.}
	Finally, we give
	evidence that our model gives rise to graph structures that  resemble real-world networks. Even when the initial network is an Erd\"{o}s-Renyi graph, measures such as the degree distribution and triangle density shift to resemble real social networks due to our edge  dynamics. E.g., while the degree distribution of the Erd\"{o}s-Renyi graph is initially binomial (approximately normal), our edge dynamics drive the network to a steady state degree distribution that appears closer to a power-law, as expected in real social networks \cite{Muchnik:2013}.  See Fig. \ref{fig:ERDeg}.  
	Such findings help to validate our model's realism, as a mechanism for opinion and network coevolution.
	
	\begin{figure}[h]
		\minipage{0.36\textwidth}
		\vspace{-0.5em}
		\includegraphics[width=\linewidth]{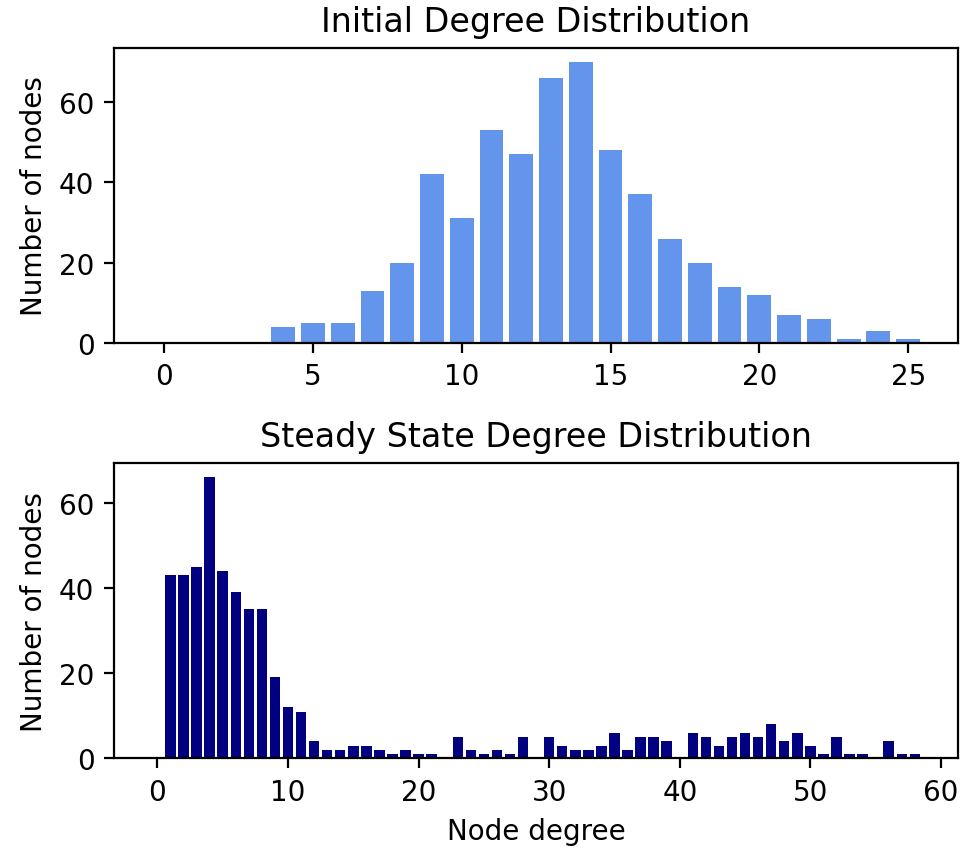}
		\vspace{-2em}
		\endminipage
		\caption{Initial and steady state degree histograms for an Erd\"{o}s-Renyi graph with $531$ nodes \& average degree $\approx 14$, subject to our edge dynamics. 
			The steady state distribution differs significantly from the initial distribution. It is closer to a power law distribution, reflecting a more realistic network structure. 	\vspace{-2em}}\label{fig:ERDeg}
		
	\end{figure}

	\subsection{Related Work}
	
	Polarization and its connection to recommender systems and opinion dynamics has seen significant research interest, particularly in the wake of Pariser's 2012 `filter bubble' hypothesis \cite{Pariser:2012}.  
	
	
	A number of papers consider the effect of link recommendations, including friend-of-friend recommendations, on network structure \cite{DalyGeyerMillen:2010,SuSharmaGoel:2016}. Others consider the effect of edge rewiring and innate opinion perturbation on polarization, including within the Friedkin-Johnsen model \cite{MuscoMuscoTsourakakis:2018,AbebeKleinbergParkes:2018,ChenRacz:2020,ChitraMusco:2020,GaitondeKleinbergTardos:2020, Racz:2022tu}. Edge rewiring can model a recommender system, an adversary that seeks to maximize polarization, or a benevolent administrator that seeks to minimize polarization. Generally, these works focus on how a single intervention can effect polarization, often finding that relatively minor changes can have significant impact.
	Unlike our work, they do not consider how opinions and the underlying network coevolve and drive polarization over time.
	
	Several works do consider opinion and network coevolution \cite{HolmeNewman:2006,CastellanoFortunatoLoreto:2009,DandekarGoelLee:2013,BaumannLorenz-SpreenSokolov:2020,SasaharaChenPeng:2020}. A common finding, matching our results, is that confirmation bias (also called homophily), itself is not  sufficient to drive significant polarization \cite{DandekarGoelLee:2013, SasaharaChenPeng:2020}. Dandekar et al. \cite{DandekarGoelLee:2013} introduce another psychological factor, \emph{biased assimilation}, in which nodes that are presented a mixture of opinions  give undue support to their initial opinion. They show that this dynamic drives polarization in combination with confirmation bias. Sasahara et al. \cite{SasaharaChenPeng:2020} present a model in which a user's expressed opinion only takes into account sufficiently agreeable neighbors. They show that this behavior 
	drives the network into a bimodal opinion distribution. They also show that direct recommendations of agreeable edges accelerate polarization. 
	
	Like our work, several related works study the validity of their synthetic opinion network models compared to real-world data \cite{SasaharaChenPeng:2020, EvansFu:2018, SuSharmaGoel:2016}.   Although specific validation methods vary, many works examine graph structures such as clustering, triangles, and degree distribution.  Sasahara et al. \cite{SasaharaChenPeng:2020} use such measures to show that their model can produce a snapshot which has similar features to real-world social network data.  Others show that these graph structures can reflect changes to the recommendation dynamics in a social network \cite{SuSharmaGoel:2016}, or correspond with parameters that change the behavior of their synthetic model \cite{EvansFu:2018}.
	
	The above works are complementary to ours. They use custom opinion dynamics models, while we build on the standard Friedkin-Johnsen model. This ties our work to well-established studies of opinion dynamics and allows us to leverage  theoretical tools from prior work. Additionally, prior work  does not consider the effect of fixed edges or friend-of-friend link recommendations, focusing instead on other important dynamics that can drive polarization.
	
	Contrasting our work, some works suggest that recommendation systems can actually mitigate filter bubble effects \cite{NguyenHuiHarper:2014,AridorGoncalves:2020} . These works focus on diversity of content consumption, showing that it can be increased by recommendations based on natural collaborative filtering. 
	%
	%
	%
	Other works propose remedies for `fixing' the polarizing effect of recommender systems, or augmenting networks to reduce polarization \cite{CelisKapoor:2019, HaddadanShahrzadMenghini:2021, Racz:2022tu, ChitraMusco:2020}.  Some of these works \cite{Racz:2022tu, ChitraMusco:2020} also work with the Friedkin-Johnsen model.
	It would be interesting to understand the effects of these remedies in our opinion and edge dynamics model.
	
	
	\section{Opinion and Edge Dynamics Model}\label{sec:prelim}

	We start by defining preliminaries and detailing our model of opinion formation under local edge dynamics.
	\vspace{-.5em}
	\subsection{Opinion Dynamics} 
	We work with the Friedkin-Johnsen opinion model \cite{FriedkinJohnsen:1990}. There are $n$ individuals, connected by an undirected graph $G$ with Laplacian matrix $L \in \R^{n \times n}$. There is an \emph{innate opinion vector} $s \in \R^n$, whose entries represent each individual's opinion without influence from neighbors.  Node opinions are numerically coded along the interval $[-1,1]$, and we assume they are drawn from a distribution with mean $0$. 
	An  \emph{expressed opinion vector} $z \in \R^n$ models the individuals' opinions under influence from their neighbors. $z$ is obtained by repeatedly applying the opinion averaging update:
	$
	z(i) := \frac{s(i) + \sum_{j} w_{ij} z(j)}{1+\sum_{j} w_{ij}},
	$
	where $w_{ij}$ is the weight of the edge between node $i$ and node $j$. For simplicity, we only consider unweighted graphs, where $w_{ij} = 0$ if there is no edge and $w_{ij} =1$ if there is an edge. This update converges to an equilibrium, with 
	$z = (I+L)^{-1} s$.  Note that since the innate opinion vector $s$ is mean-centered, $z$ will also be mean-centered, as shown in Proposition 2 from \cite{MuscoMuscoTsourakakis:2018}
	
	
	
	\vspace{-.5em}
	\subsection{Edge Dynamics}\label{sec:edgedyn}
	In our model, the network and the expressed opinions coevolve over time. The innate opinions are fixed at initialization. 
	%
	%
	%
	Let $L_t \in \R^{n \times n}$ denote the network Laplacian at time $t$, and $z_t = (I+L_t)^{-1}s$ denote the expressed opinions. At each time step, we compute a set of edges to be removed and a set of edges to be added to $G$. These sets typically depend on the expressed opinions $z_t$. After the removals and additions, $L_t$ is updated to $L_{t+1}$, and the process continues.
	
	At initialization, we  select a small percentage (typically $5\%$) of random edges to be \textit{fixed}, and so not subject to deletion. They can be thought of as modeling connections that are independent of opinions or recommendations, e.g., to family  or co-workers. 
	
	\smallskip
	
	\noindent\textbf{Edge Removals.}
	We set a percentage $p$ of edges in the graph to be removed in each step. Typically, $p = 10\%$.
	In the Appendix, in Fig. \ref{fig:conn}, we show that the choice of $p$ has little effect on the model's behavior -- a smaller or larger value of $p$ simply scales the number of time steps necessary for the edge dynamics to converge.
	
	We then select the set to be removed according to a probability distribution. In the control, this distribution is uniform over non-fixed edges. To model confirmation bias, the distribution is based on expressed opinion disagreement. A non-fixed edge $(i,j)$ is removed with probability proportional to $|z_t(i) - z_t(j)|$. In this way, nodes tend to remove connections to other nodes that express conflicting opinions, while keeping edges that confirm their own opinion.  
	
	We sample $k = \lfloor p \cdot e \rfloor$ edges to be removed, without replacement, where $e$ is the number of edges in the graph, excluding any fixed edges.  We then iterate over the sampled edges, removing each from the graph in turn. 
	
	\smallskip
	
	\noindent\textbf{Edge Additions.} 
	Given the number of edges $r$ that were removed at the current step, we select $r$ edges to be inserted. For the control, we simply sample $r$ edges that are not currently in the graph uniformly at random. To model friend-of-friend recommendations, we select $r$ edges iteratively. We select a random node and compute its friend-of-friends set -- i.e., its two-hop neighbors. We then select one of these friend-of-friends uniformly at random and add an edge between it and the source node to an `addition set'.  The process continues until there are $r$ edges in the addition set, at which point all edges in the set are added to the graph.
	
	\vspace{-.5em}
	\subsection{Polarization and Disagreement}\label{sec:pd}
	
	The primary quantities that we measure as the network $L_t$ and expressed opinions $z_t$ coevolve over time are \emph{polarization} and \emph{disagreement}. The polarization is the variance of the expressed opinions, a common choice in the literature \cite{MuscoMuscoTsourakakis:2018,AbebeKleinbergParkes:2018,BrooksPorter:2020}. The disagreement measures the variance of expressed opinions just across edges currently in the graph. It is high if nodes tend to be connected to other nodes with very different opinions, and low if nodes tend to be connected to other nodes with similar opinions.
	
	As shown in \cite{MuscoMuscoTsourakakis:2018}, in the Friedkin-Johnsen model, polarization and disagreement can be written as quadratic forms involving the network Laplacian and the innate opinion vector. For simplicity, we assume throughout that our innate opinions have mean $0$.
	
	\begin{fact}[Polarization]\label{def:p} Consider a graph with Laplacian matrix $L \in \R^{n \times n}$, along with a mean $0$ innate opinion vector $s \in \R^n$. Let $z = (I+L)^{-1}s$ be the equilibrium expressed opinion vector in the Friedkin-Johnsen model. The polarization is given by $$P(L,s) = \norm{z}_2^2 \eqdef s^T (I+L)^{-2} s.$$
	\end{fact}
	\begin{fact}[Disagreement]\label{def:d} Consider the setting of Fact \ref{def:p}. 
		The disagreement is given by $$D(L,s) \eqdef \sum_{i,j} w_{ij} \cdot (z(i)-z(j))^2 =  z^T L z = s^T (I+L)^{-1} L (I+L)^{-1} s.$$
	\end{fact}
	Due to its simple formulation as a quadratic form over $(I+L)^{-1}$, in some of our theoretical bounds, we  work with the \emph{polarization + disagreement}, introduced in \cite{MuscoMuscoTsourakakis:2018} and defined below. 
	\begin{fact}[Polarization + Disagreement]\label{def:pd} Consider the setting of Facts \ref{def:p} and \ref{def:d}. 
		The polarization + disagreement is given by
		\begin{align*}
			PD(L,s) \eqdef P(L,s) + D(L,s) = s^T (I+L)^{-1} s.
		\end{align*}
	\end{fact}
	From Fact \ref{def:p}, we can derive a simple upper bound on  polarization, which is useful in interpreting our simulation results. 
	\begin{proposition}\label{prop:var}
		Consider any Laplacian matrix $L \in \R^{n \times n}$, along with a mean $0$ innate opinion vector $s \in \R^n$. We have
		$P(L,s) \le \norm{s}_2^2.$
		I.e.,  polarization is bounded by the innate opinion variance. 
	\end{proposition}
	\begin{proof}
		Since $L$ is positive semidefinite, all eigenvalues of $I+L$ are at least $1$. Thus, all eigenvalues of $(I+L)^{-2}$ are at most $1$. Thus, using Fact \ref{def:p}, $P(L,s) = s^T (I+L)^{-2} s \le \norm{s}_2^2.$ 
	\end{proof}
	
	In our setting, the innate opinions $s$ are fixed at initialization. Thus, we typically write the polarization, disagreement and polarization + disagreement at time $t$ simply as $P(L_t)$, $D(L_t)$, $PD(L_t)$.
	
	\vspace{-.5em}
	\section{Theoretical Results}\label{sec:theo}
	
	We now present our theoretical results, which help explain how polarization and disagreement evolve under local edge dynamics.	
	\subsection{Single Edge Updates}\label{sec:single}
	
	Using the expressions for polarization and disagreement in Sec. \ref{sec:pd}, we can understand how these quantities evolve as edges are added and removed from the graph. 
	We consider a simplified setting in which just a single edge $(u_t,v_t)$ is added or removed at a time. This can be represented by the update $L_{t+1} = L_{t} \pm E_t$ where $E_t \in \R^{n \times n}$ is a rank-$1$ edge Laplacian for  $(u_t,v_t)$. That is, $E_t(u_t,u_t) = E_t(v_t,v_t)  =1$ and $E_t(u_t,v_t) = E_t(v_t,u_t) = -1$. We can also write $E_t = \chi_{u_t,v_t} \chi_{u_t,v_t}^T$ where $\chi_{u_t,v_t} \in \R^{n}$ is the edge indicator vector with a $1$ at position $u_t$ and a $-1$ at position $v_t$.
	We only allow adding an edge not currently in the graph and removing one in the graph. This ensures that $L_t$ remains a valid unweighted graph Laplacian. 
	
	
	We compute the update in polarization + disagreement under edge additions/deletions via the Sherman-Morrison formula \cite{ShermanMorrison:1950}.
	\begin{lemma}[P+D -- Edge Addition/Delete]\label{lem:add} Consider any unweighted graph Laplacian $L \in \R^{n \times n}$, and let $s,z \in \R^{n}$ be the innate and expressed opinion vectors in the Friedkin-Johnson model. Let $E \in \R^{n \times n}$ be the edge Laplacian for edge $(u,v)$. Let $\delta= z(u)-z(v)$ and $r_{u,v} = \chi_{u,v}^T (I+L)^{-1} \chi_{u,v}$.
		\begin{align*}
			PD(L + E) = PD(L) - \delta^2/(1+ r_{u,v}). \\
			PD(L - E) = PD(L) + \delta^2/(1- r_{u,v}).
	\end{align*}
\end{lemma}
We defer the proof of Lem. \ref{lem:add} to the appendix.
Note that $r_{u,v} = \chi_{u,v}^T (I+L)^{-1} \chi_{u,v}$ is the \emph{effective resistance} between $(u,v)$ in the graph given by $L$ plus a small copy of the complete graph \cite{Spielman:2019}. It will be small if $(u,v)$ are algebraically well-connected. For any $L$ and $(u,v)$ we have  $r_{u,v} \ge 0$  and $r_{u,v} \le \norm{\chi_{u,v} }_2^2 \le 2$. If  $(u,v)$ is already in the graph (as in a deletion), $r_{u,v}  \in (0,1]$. This gives:
\begin{corollary}[P+D -- Edge Addition/Deletion Bounds] \label{cor:add}
	Consider the setting of Lemma \ref{lem:add}. We have: 
	$$
	PD(L) - \delta^2 \le PD(L+E) \le PD(L) - \delta^2/3.
	$$
	$$
	PD(L-E) \ge PD(L) + \delta^2.$$
\end{corollary}

From Cor. \ref{cor:add}, we see that adding an edge decreases the polarization + disagreement. Subtracting an edge increases it. In both cases, the magnitude of change is roughly linear in the squared disagreement across the edge. We highlight that, since the disagreement is in terms of the \emph{expressed opinions}, which may differ substantially from the innate opinions, this finding is non-obvious. It is surprising that the change in polarization + disagreement only depends on the innate opinions $s$ through the expressed opinions $z$.

\vspace{-.5em}
\subsection{Edge Swaps}\label{sec:swap}

Building on the above, we next consider how $PD(L)$ changes when a disagreeable edge is swapped out for a more agreeable one. Our main result is Cor. \ref{cor:main}:  if there is a sufficient gap in disagreement across pair $(i,j)$ and pair $(k,\ell)$, then removing edge $(i,j)$ and adding $(k,\ell)$ will strictly increase the polarization + disagreement. Thus, as long as there remain edges in the graph with higher disagreement than non-edges, swapping out disagreeable connections for agreeable ones will drive up $PD(L)$. 
This finding helps explain  why a combination of confirmation bias and friend-of-friend recommendations leads to a significant increase in  polarization.
\begin{lemma}[P+D -- Edge Swap]\label{lem:swap} Consider any unweighted graph Laplacian $L \in \R^{n \times n}$, and let $s,z \in \R^{n}$ be the innate and expressed opinion vectors in the Friedkin-Johnson model.  Let $E_1$ be the edge Laplacian for edge $(u_1,v_1)$ and $\delta_1= z(u_1)-z(v_1)$. Let $E_2$ be the edge Laplacian for edge $(u_2,v_2)$ and $\delta_2= z(u_2)-z(v_2)$. Assume that the edge $(u_2,v_2)$ is in the graph corresponding to $L$. Let $r_1 = \chi_{1}^T (I+L)^{-1} \chi_{1}$, $r_2 = \chi_{2}^T (I+L)^{-1} \chi_{2}$, and $r_{2,1} = \chi_2^T (L+I+E_1)^{-1}\chi_2$.
	\begin{align*}
		PD(L+E_1-E_2) &\ge PD(L)- \frac{\delta_1^2}{1 + r_1} + \frac{(\delta_2 - \alpha \cdot \delta_1)^2}{1-r_{2,1}}, 
	\end{align*}
	where $\alpha = \frac{\chi_1^T (L+I)^{-1} \chi_2}{1+r_1}$ and $|\alpha| \le \frac{\sqrt{r_1 \cdot r_2}}{1+r_1}$.
\end{lemma}
We defer the proof of Lem. \ref{lem:swap} to the appendix. It follows by applying the two formulas of Lem. \ref{lem:add}  in sequence, and simplifying via a third application of the Sherman-Morrison formula.

In a well-connected graph,  $r_1, r_2 \ll 1$. 
hus, the increase in PD(L) will be roughly $\frac{\delta_2^2}{1-r_{2,1}} - \frac{\delta_1^2}{1+r_1} \ge \delta_2^2 - \delta_1^2$. I.e., adding a more agreeable edge and removing a more disagreeable one will increase PD. Formally, in the appendix we prove:
\begin{corollary}[PD Increase with Swap]\label{cor:main}
	Consider the setting of Lemma \ref{lem:swap}. If $|\delta_2| > \frac{3\sqrt{3}}{4} \cdot |\delta_1|$ then 
	$PD(L+E_1-E_2) > PD(L).$
\end{corollary}

\vspace{-.5em}
\subsection{Stochastic Block Model Analysis}\label{sec:sbm}

The results of Secs \ref{sec:single} and \ref{sec:swap} help explain how polarization and disagreement evolve incrementally as edges are added and removed. As discussed, we find that the evolution of these quantities at larger time scales can be well-approximated by looking simply at connectivity across opinion groups. In particular, we  approximate $L_t$ with an expected Stochastic Block Model (SBM) graph with the same in-group and out-group connectivity. We approximate the innate opinion vector $s$ with a discretization of that vector that is constant on each opinion group. Following work of \cite{McSherry:2001, ChitraMusco:2020}, we can compute closed form expressions for the polarization and disagreement for this SBM graph and innate opinion vector. We find that these  approximations closely match empirical observations -- see Fig. \ref{fig:fixedIntro}.

We start by formally defining our SBM approximation. For simplicity, we assume that the number of nodes $n$ is even -- odd $n$ is easily handled via rounding. For two vertex sets $A,B$, we let $E(A,B)$ denote the number of edges between these sets in the graph.

\begin{definition}[SBM Approximation]\label{def:sbm} Consider a graph Laplacian $L \in \R^{n \times n}$ and innate opinion vector $s \in [-1,1]^n$. Let $V_+$ and $V_-$ be the vertex sets corresponding to nodes with positive and negative opinions, respectively.
Let $q = \frac{E(V_+,V_-)}{|V_+| \cdot |V_-|/2}$ and $p = \frac{E(V_+,V_+) + E(V_-,V_-)}{(|V_+|^2 + |V_-|^2)/2}$ be the fraction of out-group and in-group edges, respectively. 

Let $\bar L$, be expected SBM graph with out-group and in-group connection probabilities $q,p$. In particular, $\bar L = \frac{(p+q)n}{2} I - \bar A$ where  $\bar A(i,j) = \bar A(j,i) = p$ if $i,j \in [1, n/2 ] $ or $i,j \in [n/2+1, n]$. $\bar A(i,j) = \bar A(j,i) = q$ if $i \in [1,n/2]$ and $j \in [n/2,n]$.
Let $\bar s$ have $\bar s(i) = 1$ for $i \in [1,n/2]$ and $\bar s(i) = -1$ for $i \in [n/2+1,n]$.
\end{definition}
Observe that opinion groups in the expected SBM graph $\bar L$ correspond to the first and last $n/2$ nodes. If $s$ is chosen randomly according to a symmetric distribution,  $|V_+| \approx |V_-| \approx n/2$. Thus, it is reasonable to approximate both groups as having size exactly $n/2$. Also observe that $\bar A$ is simply a $2 x 2$ block matrix, with its top-left and bottom-right $n/2 \times n/2$ blocks filled with $p$'s and its top-right and bottom-left blocks filled with $q$'s. We can check that $\bar s$ is an eigenvector of $\bar A$ with eigenvalue $ \frac{(p-q)n}{2}$. Thus, $\bar s$ is an eigenvector of $\bar L = \frac{(p+q)n}{2} I - \bar A$ with eigenvalue $qn$. Using this, we  derive:
\begin{fact}[SBM Polarization and Disagreement]\label{fact:sbm}
Let $\bar L \in \R^{n \times n}$ and $\bar s \in \R^n$ be as defined in Def. \ref{def:sbm}. We have:
	\begin{align*}		
	 P(\bar L, \bar{s}) = \frac{n}{(qn + 1)^2}\hspace{1em}\text{ and }\hspace{1em}
		D(\bar L, \bar{s}) = \frac{qn^2}{(qn + 1)^2}.
	\end{align*}
\end{fact}

\begin{proof}
	Using Facts \ref{def:p} and \ref{def:d}, $P(\bar L, \bar s) = \bar s^T (I + \bar L)^{-2} \bar s$ and $D(\bar L,\bar s) = \bar s^T (I+\bar L)^{-1} \bar L (I+\bar L)^{-1} \bar s$.  Since $\bar s$ is an eigenvector of $\bar L$ with eigenvalue $qn$, it is an eigenvector of both  $(I + \bar L)^{-2}$ and $(I+\bar L)^{-1} \bar L (I+\bar L)^{-1}$ with eigenvalues $\frac{1}{(qn+1)^2}$ and $\frac{qn}{(qn+1)^2}$, respectively. The fact then follows from observing that $\norm{\bar s}_2^2 = n$.
\end{proof}

Via Fact \ref{fact:sbm}, we can approximate to polarization and disagreement of a graph $L$ and innate opinion vector $s$ with a simple formula that depends just on the out-group connection density $q$. Surprisingly, the formula \emph{has no dependence} on the in-group density $p$. Since we always have $\norm{\bar s}_2^2 = n$, we scale the quantities in Fact \ref{fact:sbm} by a factor of $\frac{\norm{s}_2^2}{n}$ to adjust for the innate opinion variance. The resulting approximations are quite accurate in predicting  polarization and disagreement evolution in our model -- see Fig. \ref{fig:fixedIntro}.  

\smallskip

\noindent{\textbf{Interpretation.} 
The accuracy of our SBM-based approximation indicates that the evolution of polarization and disagreement in our edge dynamics model is largely governed by the evolution of out-group connection probabilities. We observe that a combination of confirmation bias and friend-of-friend recommendations tends to drive down $q$ over time and hence drive up polarization. When either dynamic is removed, $q$ remains relatively high, and polarization does not rise significantly. See Fig. \ref{fig:compIntro2}

Interestingly, our disagreement approximation $D(\bar L, \bar{s}) = \frac{qn^2}{(qn + 1)^2}$ is not a monotonic function of $q$. $D(\bar L, \bar{s}) $ is decreasing in $q$ when $q \ge 1/n$, but increasing when $q \le 1/n$. This explains an interesting phenomena seen in our model with no fixed edges: the disagreement initially `spikes' as out-group connections are removed, before falling to near $0$, as the opinion groups become fully disconnected (i.e., as $q$ becomes very small). When we fix a $\gamma$  fraction of fixed edges, $q$ is effectively lower bounded by $\gamma$. As long as $\gamma \ge 1/n$, we never see disagreement fall -- it rises jointly with polarization. See Fig. \ref{fig:fixedIntro}.
Relatedly, when $q$ is not lower bounded by $\gamma$, $P(\bar L, \bar s)$ is able to reach its maximum value of $n$, giving a predicted polarization of $\norm{s}_2^2$, which matches the maximum  polarization  derived in Prop. \ref{prop:var}. 

	
	
%

\begin{figure}[h]
	\vspace{-1em}
	\minipage{0.45\textwidth}
	\includegraphics[width=\linewidth]{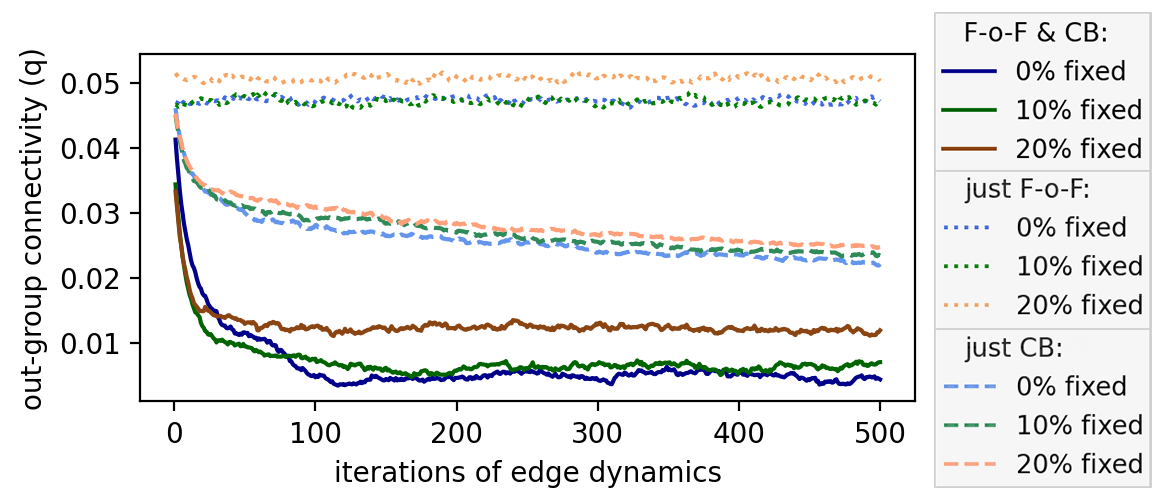}
	\endminipage
	\vspace{-1em}
	\caption{Out-group connectivity $q$ over time for Erd\"{o}s-Renyi graphs with $1000$ nodes and connection probability $p = 0.05$, with varying percentages of fixed edges.  Innate opinions are uniformly distributed in $[-1,1]$. With edge dynamics influenced by both \textit{confirmation bias} (CB) and \textit{friend-of-friend recommendations} (F-o-F), out-group connectivity drops significantly, mirroring a rise in polarization, as predicted by  Fact \ref{fact:sbm} and observed in Figs. \ref{fig:compIntro} and \ref{fig:fixedIntro}. 
		\vspace{-1em}}\label{fig:compIntro2}
\end{figure}

	\section{Experimental Results}\label{sec:exp}
	
	We next discuss our experimental results simulating our opinion and edge dynamics model on synthetic and real world data.
	
	\subsection{Experimental Set Up}\label{sec:expsetup}
	
	\noindent\textbf{Synthetic Networks.}
	We simulate our model on several random graphs, drawn from the Erd\"{o}s-Renyi (ER) and Barab\'{a}si-Albert (BA) distributions. ER graphs have each pair of nodes is connected independently with probability $p$. We modulate this probability in our experiments to study the effects of network density. BA graphs are scale-free networks generated via a preferential attachment process. For these graphs, we modulate the primary parameter $m$, which determines the number of edges to attach from a new node to existing nodes as the graph is generated.
	
	\smallskip
	
	\noindent\textbf{Real-world Networks.}
	We also simulate our model on several real-world networks. We preprocess all networks by computing their 2-core, which is a maximal subgraph such that each node has degree at least two. This eliminates the large number of single-edge nodes present in some of these datasets, which would distort the results.  Basic graph parameters (after preprocessing) are listed below, with more details on the datasets given in Appendix \ref{app:data}.
	\begin{itemize}
		\item \textbf{Reddit }\cite{DeAbir:2014} \;\; \textit{$n = 546$ nodes, $e = 8962$ edges}
		\item \textbf{Twitter }\cite{DeAbir:2014} \;\; \textit{$n = 531$ nodes, $e = 3621$ edges}
		\item \textbf{Facebook }\cite{LeskovecJureMcauley:2012} \;\; \textit{$n = 3964$ nodes, $e = 88159$ edges}
	\end{itemize}
	%
	
	\smallskip
	
	\noindent\textbf{Innate Opinions.} In our primary experiments, innate opinions are sampled uniformly at random from the interval $[-1,1]$.  In the appendix we also give results for a bimodal truncated Gaussian innate opinion distribution, observing similar behavior.
	
	\smallskip
	
	\noindent{\textbf{Trials.} For all experiments, we report the average behavior over five independent trials.
	
	
	
	

\subsection{Drivers of Polarization}\label{sec:happens}

As discussed, our edge dynamics with both confirmation bias and friend-of-friend recommendations drive a large increase in polarization. 
Without fixed edges, this finding is robust to the initial graph size and structure, and to the innate opinion distribution. See Figs. \ref{fig:sizes} and \ref{fig:fixed2} in the appendix, which show similar behavior for large random graphs and real social network graphs, respectively. 
See Fig. \ref{fig:bimodalPL} for an illustration with a bimodal  innate opinion distribution.  When fixed edges are present, they limit and slow the increase in polarization -- this effect is amplified for larger fractions of fixed edges.

Isolating the effects of friend-of-friend recommendations and confirmation bias, we find that replacing either with a random control drastically changes the behavior of polarization.  See Fig. \ref{fig:compIntro}
for an illustration on an ER random graph with 1000 nodes, the Twitter real-world data set \cite{DeAbir:2014}, and a dense ER random graph with 1000 nodes.
Polarization is essentially constant with just friend-of-friend recommendations, and it rises at a very slow rate with just confirmation bias.  Note that polarization in the dense ER graph does not rise significantly, even with both friend-of-friend recommendations and confirmation bias. We discuss this finding  below.

\begin{figure}[h]
	\vspace{-0.5em}
	\minipage{0.47\textwidth}
	\includegraphics[width=\linewidth]{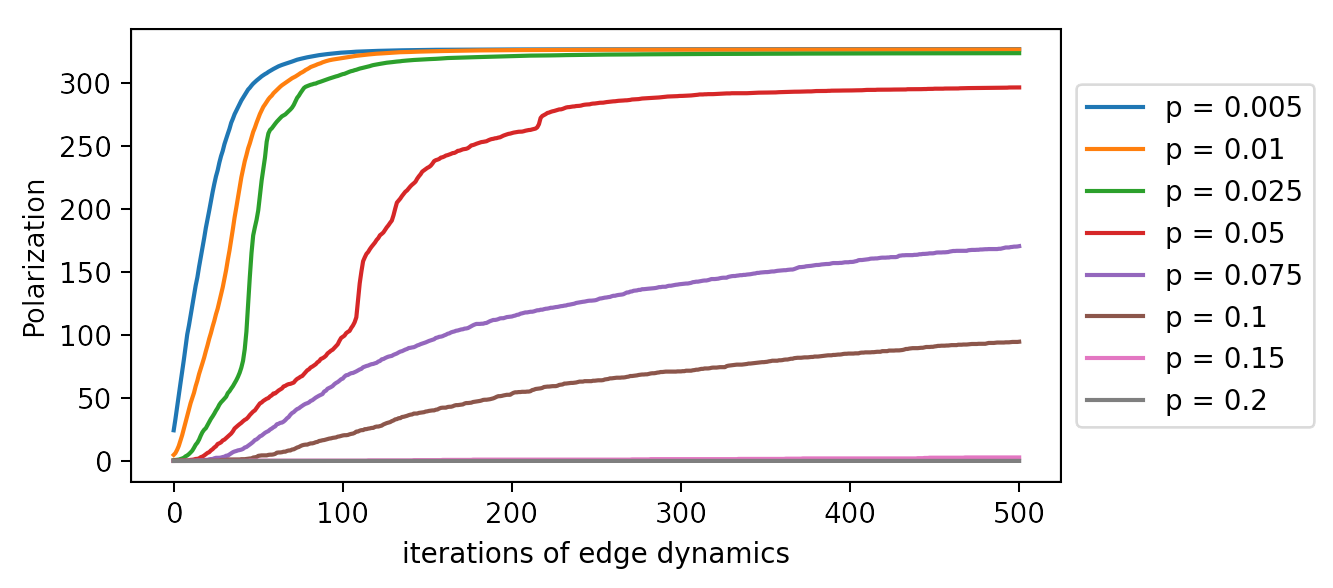}
	\endminipage
	\vspace{-0.5em}
	\caption{Polarization for ER random graphs with $1000$ nodes, no fixed edges, and varied connection probabilities. Dynamics include both friend-of-friend recommendations and confirmation bias. We see polarization start to decrease once the density passes a high enough threshold, near $1 / \sqrt{n} \approx 0.03$.}\label{fig:exp3}
	\vspace{-1.5em}
\end{figure}

\subsubsection{Effect of Density}\label{sec:density}
For particularly dense graphs, the asymptotic behavior of our model with both friend-of-friend recommendations and confirmation bias changes, as observed in Fig. \ref{fig:exp3}. 
We see that denser graphs exhibit lower polarization when the simulation is complete, and that polarization also increases at a slower rate.  

For ER graphs, a change in behavior seems to occur roughly when the average degree $d_{avg} = pn$ is large enough such that the friends-of-friends set with size roughly $d_{avg}^2$ encompasses nearly all nodes. I.e., when $d_{avg}^2 = p^2 n^2 \gtrsim n$, and so $p \gtrsim 1/\sqrt{n}$. In this case, friend-of-friend recommendations at initialization are essentially uniformly random, mirroring our control setting. A more precise theoretical understanding of this behavior would be very interesting and would help clarify the importance of friend-of-friend recommendations in amplifying the polarizing effects of confirmation bias.


\vspace{-.5em}
\subsection{Evolution of Polarization \& Disagreement}\label{sec:evolves}

For both synthetic graphs and real-world networks, using friend-of-friend recommendations and confirmation bias, we find that expressed opinions evolve through roughly three distinct states.  
\smallskip

\noindent\textbf{Initial state:} Opinions converge to be very close to the mean of the innate opinions -- which is $\approx 0$ in our setting. 
Polarization and disagreement are low.  The network's connectivity within similar opinions and between differing opinions is roughly equal.

\smallskip

\noindent\textbf{Bimodal polarization:} Expressed opinions bifurcate, "pulling apart" into a few distinct clusters on either side of the mean.  The network's connectivity is strengthening between similar opinions, and eroding between differing opinions (see Fig. \ref{fig:compIntro2}).  Both polarization and disagreement rise, with disagreement reaching a steady state if the network has fixed edges and a peak if it does not, as predicted by the SBM analysis of Sec. \ref{sec:sbm}.  

\smallskip

\noindent\textbf{Maximal polarization:} If fixed edges are not present in the network, the edge dynamics finally ``splinter'' the network into many components.  Out-group connectivity $q$ is near zero.  Polarization is near maximal as defined in Prop. \ref{prop:var} and disagreement is near zero, as predicted by the SBM analysis of Sec \ref{sec:sbm}.  Nodes cluster into a structure which keeps expressed opinions very close to innate opinions -- in the appendix, Fig. \ref{fig:mse} illustrates this by plotting the mean squared error between expressed opinions and innate opinions over time for a few experiment settings.

\smallskip 


\begin{figure}[h]
	\minipage{0.4\textwidth}
	\vspace{-1em}
	\includegraphics[width=\linewidth]{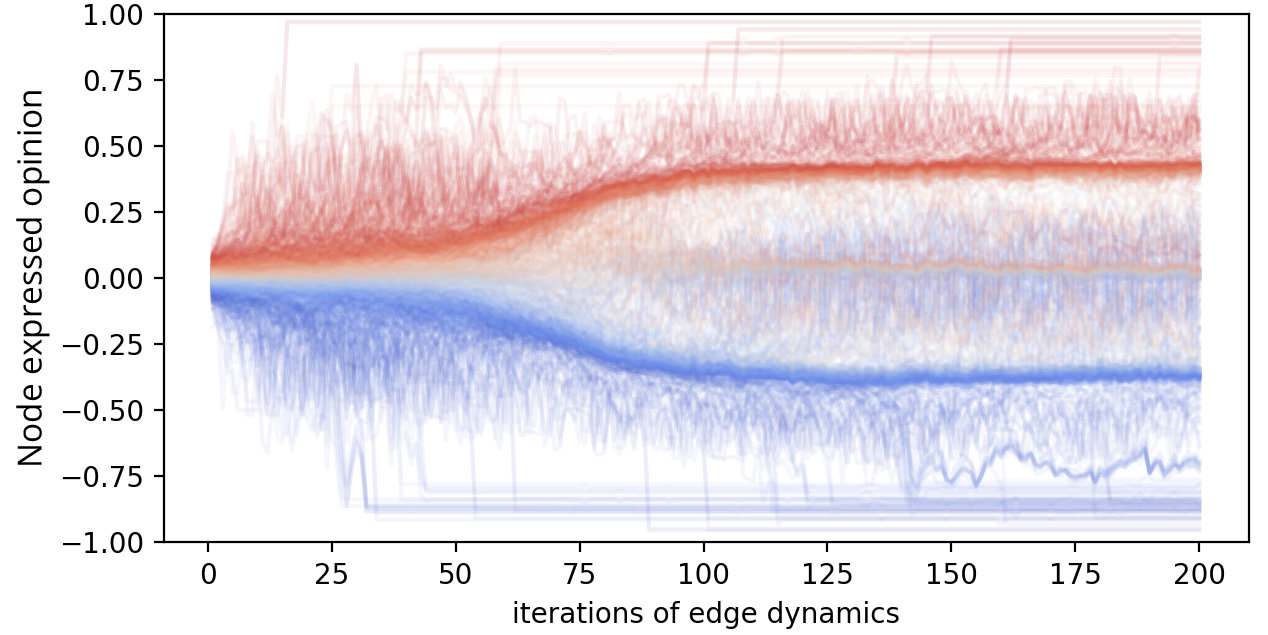}
	\endminipage
	\caption{Evolution of opinions on a Barab\'{a}si-Albert graph with $1000$ nodes, 5\% fixed edges, and uniform random innate opinions.  Each line represents the evolving expressed opinion of one node. The color of each  line represents that node's innate opinion, ranging on a gradient from -1 (blue) to 1 (red). 
		We see opinions initially grouped near $0$ (\textbf{Initial state}). %
		We then see a bifurcation of the major opinion clusters, indicated  by the red and blue ``branches'' (\textbf{Bimodal polarization}).  The fixed edges in the network prevent it from fully splintering, so nodes are ``stuck'' in this bimodal state.}\label{fig:branch}
	\vspace{-1em}
\end{figure}

In Figure \ref{fig:branch}, we show this progression of opinion states for a Barab\'{a}si-Albert graph with 5\% fixed edges.  
In Fig. \ref{fig:branch2} in the appendix, we show a similar plot for a network with no fixed edges, which continues onto the maximal polarization stage.  Appendix Figure \ref{fig:hist1} also shows histograms of node opinions which provide another visual dimension for this state evolution.   

\begin{figure*}[h]
	\vspace{-1em}
	\minipage{0.32\textwidth}
	\includegraphics[width=\linewidth]{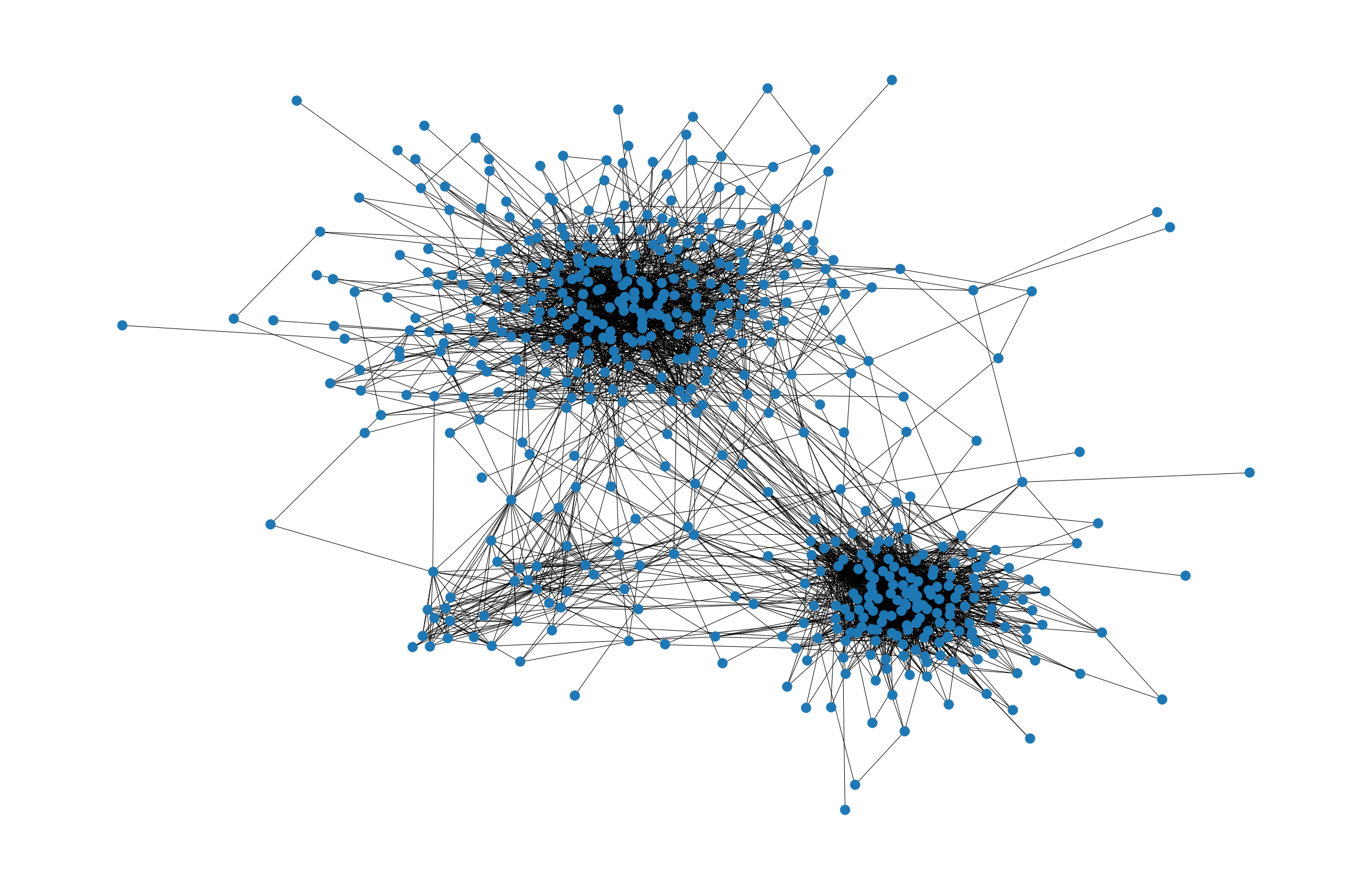}\vspace{-1em}
	\endminipage\hfill
	\minipage{0.32\textwidth}%
	\includegraphics[width=\linewidth]{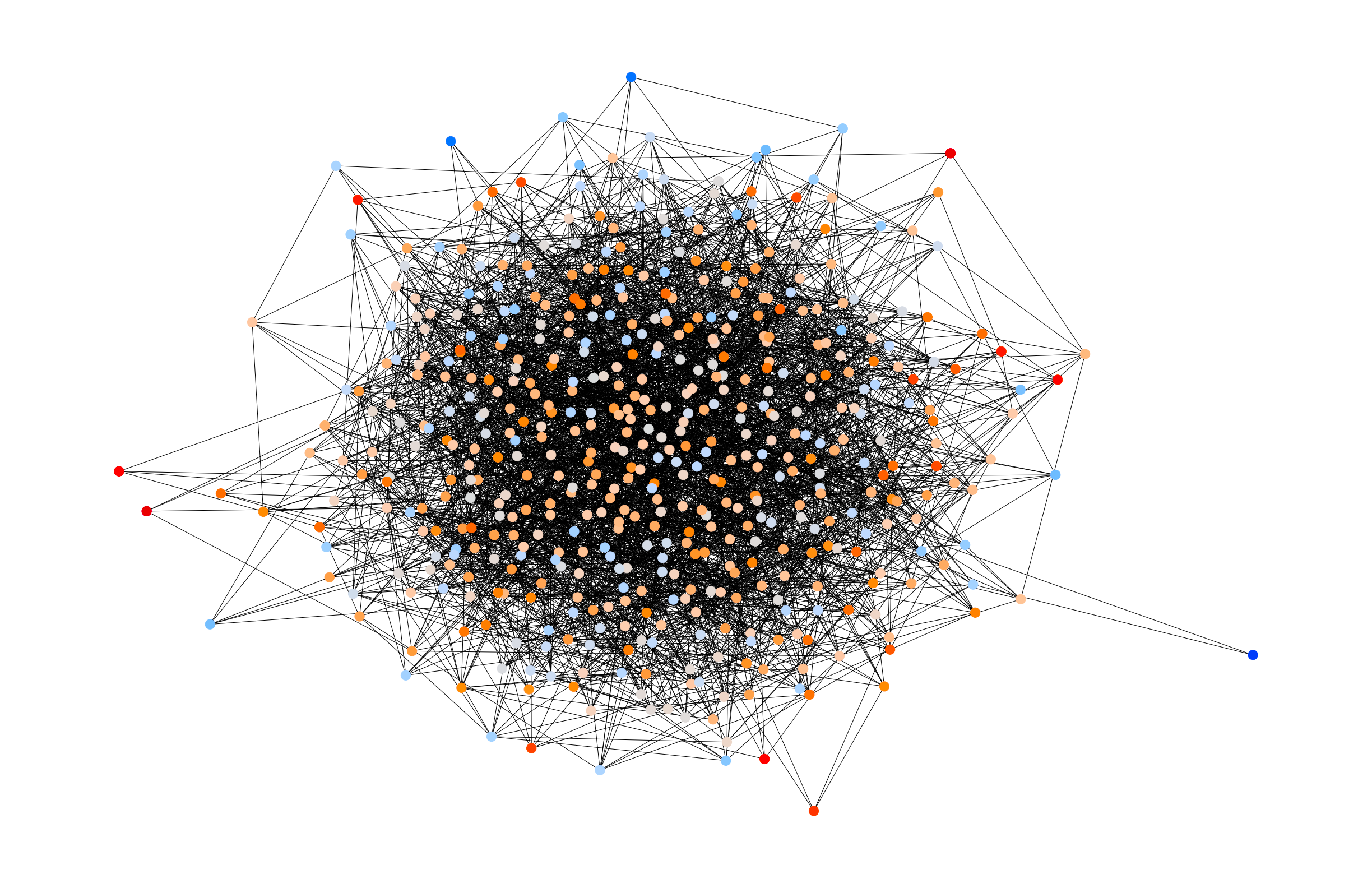}\vspace{-1em}
	\endminipage\hfill
		\minipage{0.32\textwidth}
	\includegraphics[width=\linewidth]{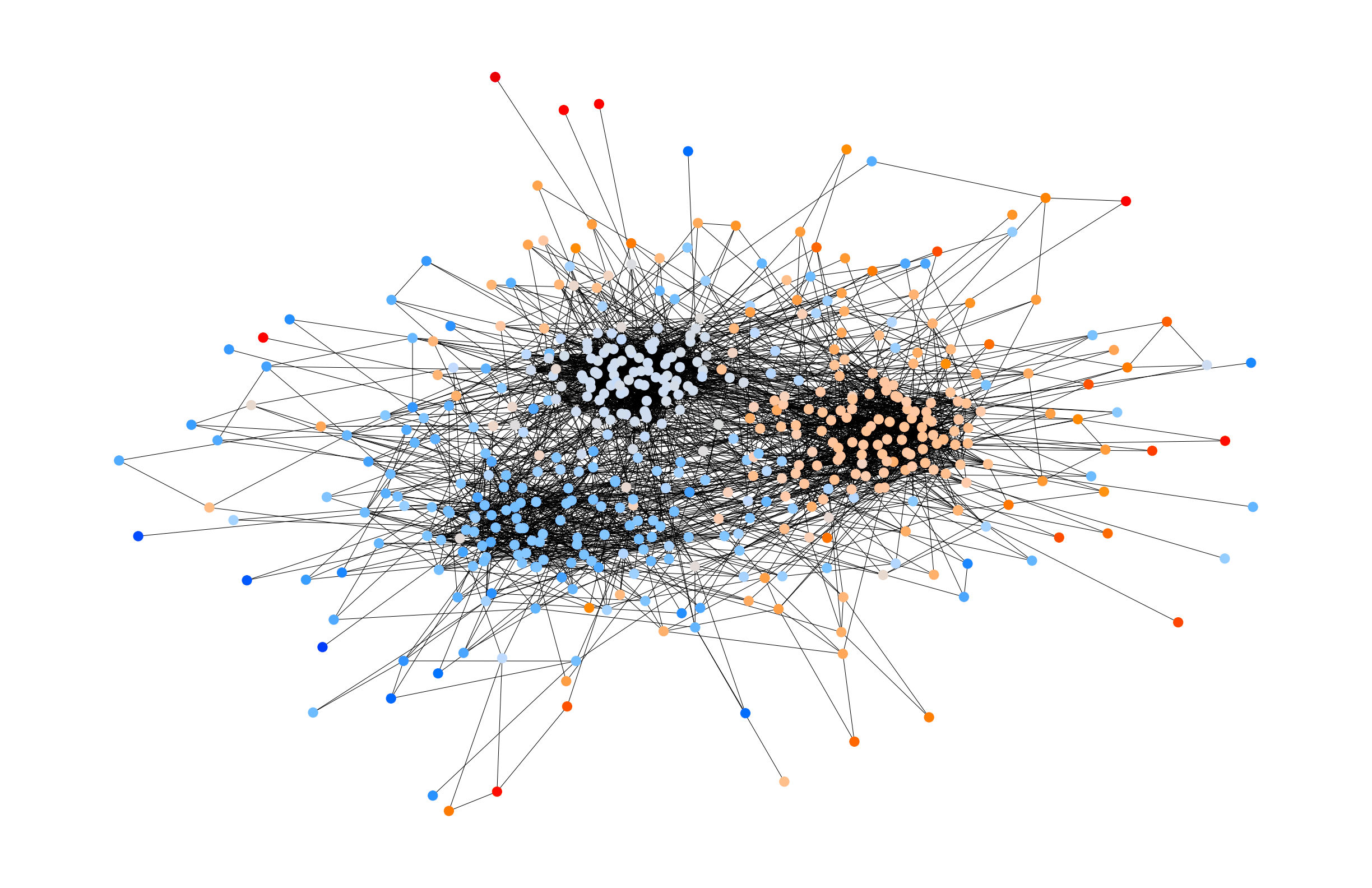}\vspace{-1em}
	\endminipage\hfill
	\caption{\textit{(left)}: Twitter \cite{DeAbir:2014} real-world snapshot. \textit{(middle)}: \textit{Initial} ER graph with 25\% fixed edges. \textit{(right)}: \textit{Steady state} ER graph with 25\% fixed edges, after edge dynamics simulation.}\label{fig:twitDelhi}
	\vspace{-.5em}
\end{figure*}

\vspace{-.5em}
\subsection{Validation Against Real-World Data}\label{sec:valid}

Following prior work \cite{SasaharaChenPeng:2020, EvansFu:2018}, in this section, we use social network datasets as a ``benchmark'' to ascertain whether our synthetic model can recreate realistic social network structures. We show validation against the Twitter data set, with results on the Facebook data set appearing in the appendix.

\smallskip

\noindent{\textbf{Set up.} 
We start with a synthetic graph (ER or BA) with the same number of nodes as the real-world network, and roughly the same number of edges.
%
We fix a certain percentage of edges (e.g., 15\%, 25\%, and 35\%), which are not subject to our edge dynamics. 

We then simulate our edge dynamics with friend-of-friend recommendations and confirmation bias, running for $1500$ iterations, until the network reaches a steady state. We measure the evolution of various structural properties of the network over time, and compare them to those of the real-world network. 

We find that the percentage of fixed edges in a graph correlates with the steady state \textit{global clustering coefficient}, which is defined as 
$\frac{3 \times \text{number of triangles in } G}{\text{number of open \& closed triads in } G}$.
Using this quantity, we tune the percentage of fixed edges in our synthetic graph.  E.g., for the Twitter network, the global clustering coefficient is $\approx 0.227$. The percentage of fixed edges leading to the closest steady-state clustering coefficient for both ER and BA  networks is $25\%$.

\smallskip 

\noindent\textbf{Degree Distribution.}
As in many social networks, Twitter has a power law degree distribution  \cite{Muchnik:2013}. See Fig. \ref{fig:twitDeg} in the appendix. 
For an ER  network with $25\%$ fixed edges, we show the initial degree distribution, along with the steady state degree distribution after running edge dynamics in Fig. \ref{fig:ERDeg}.  The initial ER graph has a bell-shaped binomial degree distribution. Surprisingly, our edge dynamics change this distribution significantly, and the steady state distribution appears closer to a  power law distribution. 

In Fig. \ref{fig:opinionDeg}, we compare each node's equilibrium degree against its expressed opinion in the steady state network.  It seems that nodes with near-mean expressed opinions \textit{gain} connections under our edge dynamics model, while peripheral nodes loose connections.  This behavior may drive the observed degree distribution shift. An improved theoretical understanding would be very interesting.

We show similar results for a BA network with $25\%$ fixed edges in the appendix,  Fig. \ref{fig:BAvalidDeg}.  Note that the BA graph  initially has a power-law degree distribution, so these results serve to validate that our model \emph{preserves} a realistic degree distribution. 

\begin{figure}[h]
	\minipage{0.4\textwidth}
	\vspace{-.5em}
	\includegraphics[width=\linewidth]{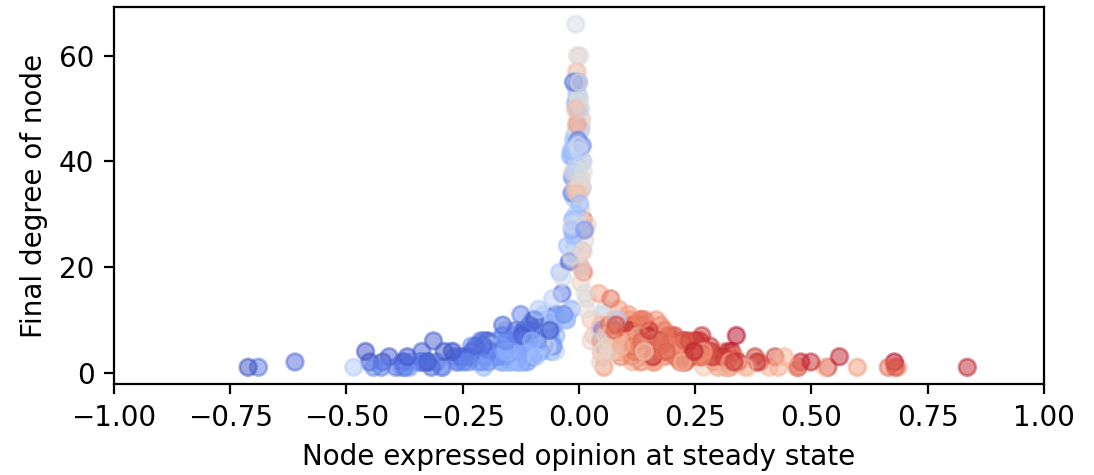}
	\endminipage
	\vspace{-1em}
	\caption{Scatter plot comparing expressed opinions and degree for each node in the steady state ER graph with 25\% fixed edges.  The color of each circle corresponds to the node's innate opinion, ranging on a gradient from -1 (blue) to 1 (red).}\label{fig:opinionDeg}
	\vspace{-1em}
\end{figure}


\smallskip

\noindent\textbf{Triangle Distribution.}
Similar to the degree distribution, in the Twitter network, the triangle distribution is power law like -- see Fig. \ref{fig:twitTri}. In Fig. \ref{fig:ERTri}, we show that our edge dynamics drive the triangle distribution of an initial ER graph to be much closer to this power law distribution. 
%
%
%
We show results for the BA generated network with $25\%$ fixed edges in the Appendix, in Fig. \ref{fig:BAvalidTri}.  Again, in this case, our dynamics preserve the already relatively realistic triangle distribution of the initial BA graph. 

\begin{figure}[h]
	\minipage{0.4\textwidth}
	\includegraphics[width=\linewidth]{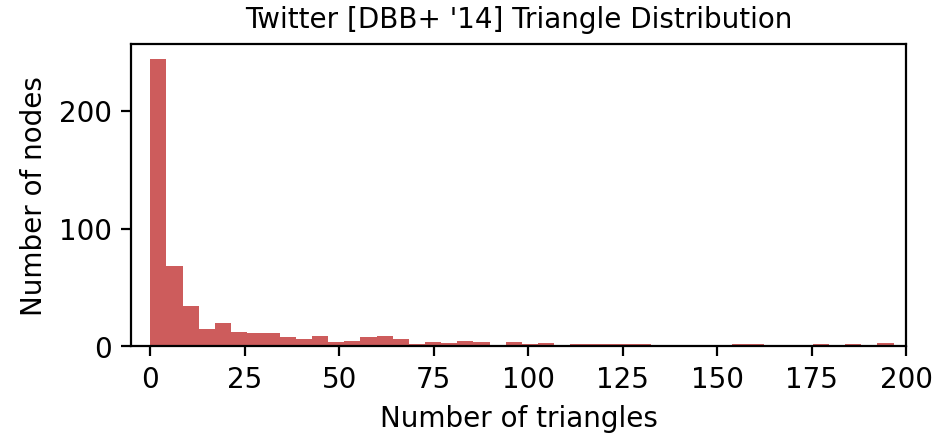}
	\endminipage
	\vspace{-1em}
	\caption{Triangle distribution for Twitter data set \cite{DeAbir:2014}.  
	}\label{fig:twitTri}
\end{figure}

\begin{figure}[h]
	\minipage{0.4\textwidth}
	\vspace{-1em}
	\includegraphics[width=\linewidth]{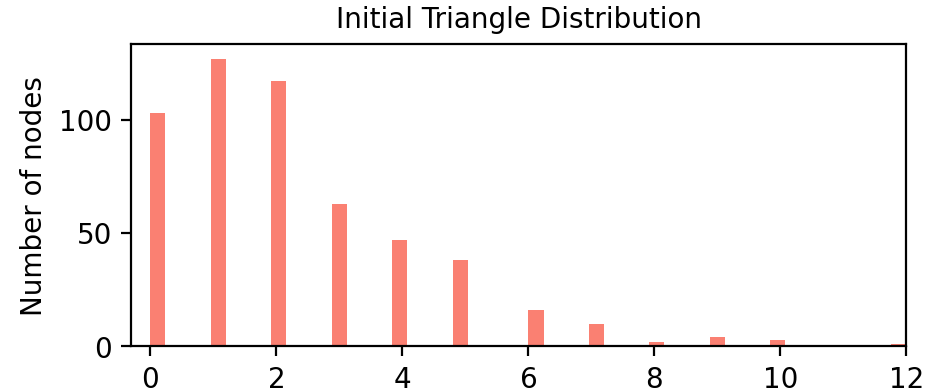}
	\endminipage\\
	\minipage{0.4\textwidth}
	\includegraphics[width=\linewidth]{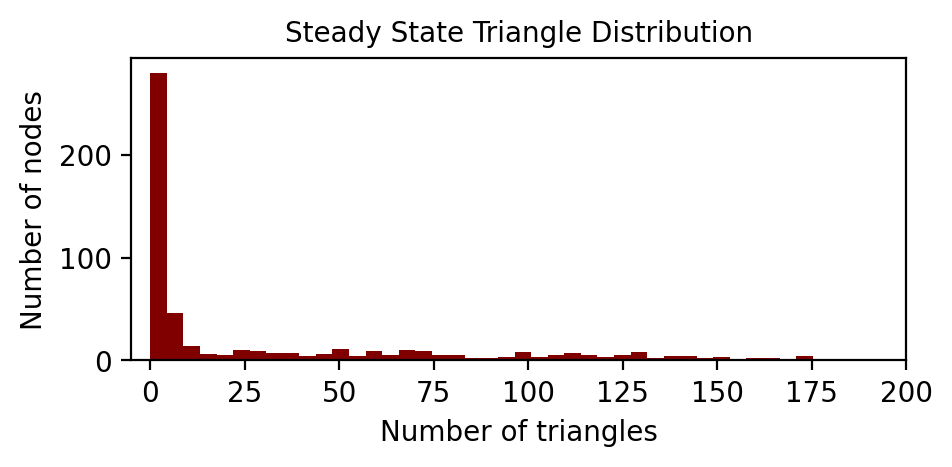}
	\endminipage
	\vspace{-1em}
	\caption{Initial and steady state triangle distributions for an ER network, subject to our edge dynamics with $25\%$ fixed edges. Notably, the steady state distribution appears closer to the Twitter network's triangle distribution (Fig. \ref{fig:twitTri}).}\label{fig:ERTri}
	\vspace{-1em}
\end{figure}

\smallskip

\noindent{\textbf{Visual Similarity.}}
While not rigorous as a method of comparison, in Fig. \ref{fig:twitDelhi} we visualize the Twitter network, an initial ER graph, and the steady-state ER graph after undergoing our edge dynamics. 
The color of each node is set according to the node's Friedkin-Johnsen innate opinion, which is uniformly sampled from the interval $[-1, 1]$.  
%
%
While the ER network's steady state connectivity between clusters seems to be more dense than the Twitter network, the overall structure appears similar, particularly around the periphery nodes and in-group clusters.  In the Appendix, in Figs. \ref{fig:baInit} \& \ref{fig:baConverged}, we show a similar visual evolution for a BA network.

\vspace{-1em}
\section{Conclusion}

We  present a simple extension of the Friedkin-Johnsen opinion model, in which the opinions and the underlying network coevolve under the influence of confirmation bias and recommendations.  Via simulation, we find that \emph{both} confirmation bias and friend-of-friend recommendations are required to noticeably increase polarization.  We show theoretical results that explain how polarization and disagreement are increased via swaps of more agreeable edges for more disagreeable ones. We also show how polarization and disagreement can be accurate approximated as functions of connectivity between opinion groups, based on a stochastic block model analysis. Finally, we validate that our opinion dynamics tend to create relatively realistic looking networks in terms of structure. 

Our findings leave open several questions.  Theoretically, better explaining the role of density in our model, and the mechanisms behind degree and triangle distribution shift would be very interesting. One also might consider whether our findings generalize to variations of the Friedkin-Johnsen model or other opinion models \cite{Degroot:1974,HegselmannKrauseothers:2002, EvansFu:2018, SasaharaChenPeng:2020}, multidimensional opinions, or network constraints beyond fixed edges.  It would be very interesting to further validate our model's realism, by extending our analysis to other graph measures, data sets with ground truth opinion data, and ideally to temporal data sets showing the coevolution of network structure and opinions in the real world.

\vspace{-1em}

\bibliographystyle{alpha}
\bibliography{polarization}


\clearpage

\appendix

\section{Deferred Proofs}

We now prove Lemmas \ref{lem:add} and \ref{lem:swap} and Corollary \ref{cor:main}, which bound the change in polarization+disagreement when adding or deleting an edge, and when swapping two edges respectively.

\begin{replemma}{lem:add} Consider any unweighted graph Laplacian $L \in \R^{n \times n}$, and let $s,z \in \R^{n}$ be an innate and corresponding expressed opinion vectors in the Friedkin-Johnsen model. Let $E \in \R^{n \times n}$ be the edge Laplacian for edge $(u,v)$. Let $\delta= z(u)-z(v)$ and $r_{u,v} = \chi_{u,v}^T (I+L)^{-1} \chi_{u,v}$.
	\begin{align*}
		PD(L + E) = PD(L) - \frac{\delta^2}{1+ r_{u,v}}.\\
		PD(L - E) = PD(L) - \frac{\delta^2}{1- r_{u,v}}.
	\end{align*}
\end{replemma}
\begin{proof}
	By the Sherman-Morrison Formula:
	\begin{align*}
		s^T (I+L + E)^{-1}s &= s^T (I+L)^{-1} s - \frac{s^T (I+L)^{-1} E (I+L)^{-1} s}{1+ \chi_{u,v}^T (I+L)^{-1} \chi_{u,v}}\\
		&= PD(L) - \frac{z^T E z}{1 + \chi_{u,v}^T (I+L)^{-1} \chi_{u,v}}.\\
		s^T (I+L - E)^{-1}s &= s^T (I+L)^{-1} s + \frac{s^T (I+L)^{-1} E (I+L)^{-1} s}{1- \chi_{u,v}^T (I+L)^{-1} \chi_{u,v}}\\
		&= PD(L) + \frac{z^T E z}{1 - \chi_{u,v}^T (I+L)^{-1} \chi_{u,v}}.
	\end{align*}
	Finally, note that $z^T E z = [z(u)-z(v)]^2 = \delta^2$.
\end{proof}

\begin{replemma}{lem:swap}
	Consider any unweighted graph Laplacian $L \in \R^{n \times n}$, and let $s,z \in \R^{n}$ be the innate and corresponding expressed opinion vectors in the Friedkin-Johnsen model.  Let $E_1$ be the edge Laplacian for edge $(u_1,v_1)$ and $\delta_1= z(u_1)-z(v_1)$. Let $E_2$ be the edge Laplacian for edge $(u_2,v_2)$ and $\delta_2= z(u_2)-z(v_2)$. Assume that the edge $(u_2,v_2)$ is in the graph corresponding to $L$. Let $r_1 = \chi_{1}^T (I+L)^{-1} \chi_{1}$, $r_2 = \chi_{2}^T (I+L)^{-1} \chi_{2}$, and $r_{2,1} = \chi_2^T (L+I+E_1)^{-1}\chi_2$.
\end{replemma}
\begin{proof}
	We apply Lemma \ref{lem:add} in sequence to give:
	\begin{align*}
		PD(L+E_1-E_2) = PD(L) - \frac{\delta_1^2}{1 + r_1} + \frac{s^T (L+I+E_1)^{-1} E_2 (L+I+E_1)^{-1} s}{1-\chi_2^T (L+I+E_1)^{-1}\chi_2}.
	\end{align*}
	We then expand out the third term using Sherman-Morrison again:
	\begin{align*}
		PD(L+E_1-E_2) &= PD(L) - \frac{\delta_1^2}{1 + r_1} + \frac{\delta_2^2}{1-r_{2,1}} \\ &+ \frac{s^T (L+I)^{-1} E_1 (L+I)^{-1} E_2 (L+I)^{-1} E_1 (L+I)^{-1}s }{(1-r_{2,1}) \cdot (1+r_1)^2}\\ &- \frac{2 s^T (L+I)^{-1} E_1 (L+I)^{-1} E_2 (L+I)^{-1}s}{(1-r_{2,1}) \cdot (1+r_1)}.
	\end{align*}
	Letting $\alpha = \frac{\chi_1^T (L+I)^{-1} \chi_2}{1+r_1}$, we can simplify the above to:
	\begin{align*}
		PD(L+E_1-E_2) &= PD(L) - \frac{\delta_1^2}{1 + r_1} + \frac{\delta_2^2}{1-r_{2,1}} + \frac{\alpha^2 \cdot \delta_1^2}{1-r_{2,1}} - \frac{2  \alpha \delta_1 \delta_2}{1-r_{2,1}}\\
		&= PD(L) - \frac{\delta_1^2}{1 + r_1} + \frac{(\delta_2 -\alpha \delta_1)^2}{1-r_{2,1}},
	\end{align*}
	completing the bound. Finally, note that since $L+I$ is positive definite, we can bound $|\alpha|$ using Cauchy-Schwarz by:
	\begin{align*}
		|\alpha| \le \frac{\sqrt{\chi_1^T (L+I)^{-1} \chi_1 \cdot \chi_2^T (L+I)^{-1} \chi_2}}{1+r_1} = \frac{\sqrt{r_1\cdot r_2}}{1+r_1}.
	\end{align*}
\end{proof}

\begin{repcorollary}{cor:main}
	Consider the setting of Lemma \ref{lem:swap}. If $|\delta_2| > \frac{3\sqrt{3}}{4} \cdot |\delta_1|$ then 
	$PD(L+E_1-E_2) > PD(L).$
\end{repcorollary}
\begin{proof}
	First note that since we assume edge $(u_2,v_2)$ is in the graph, $r_{2,1} \le r_2 < 1$. Thus by Lemma \ref{lem:swap} we have:
	\begin{align*}
		PD(L+E_1-E_2) > PD(L) - \frac{\delta_1^2}{1+r_1} + \left (|\delta_2|-\frac{\sqrt{r_1}}{1+r_1} \cdot |\delta_1| \right )^2.
	\end{align*}
	Writing $|\delta_2| = \gamma \cdot |\delta_1|$ and solving for $PD(L+E_1-E_2) = PD(L)$,
	\begin{align*}
		\frac{\delta_1^2}{1+r_1} = \left (\gamma-\frac{\sqrt{r_1}}{1+r_1} \right )^2 \cdot \delta_1^2.
	\end{align*}
	Solving for $\gamma$ under the constaint that $\gamma > 0$ gives $\gamma = \frac{\sqrt{r_1}}{1+r_1}+ \frac{1}{\sqrt{1+r_1}}.$ $\gamma$ is maximized for all $r_1 \ge 0$ at $r_1 = 1/3$ and $\gamma = \frac{3\sqrt{3}}{4}$, which completes the proof.
\end{proof}

\section{Dataset Details}\label{app:data}

\begin{itemize}
	\item \textbf{Reddit }\cite{DeAbir:2014}\\
	\textit{$n = 546$ nodes, $e = 8962$ edges}\\
	In this dataset, nodes (users) have an edge between them if there exist two subreddits in which both users posted during a given time period.  Files for this data set were obtained from previous work that cites the original source \cite{MuscoMuscoTsourakakis:2018}.
	\item \textbf{Twitter }\cite{DeAbir:2014}\\
	\textit{$n = 531$ nodes, $e = 3621$ edges}\\
	In this dataset aimed at analyzing discourse around the Delhi legislative elections of 2013, edges represent user interactions on the Twitter platform, discerned with the use of topical hashtags.  Files for this data set were obtained from previous work that cites the original source \cite{MuscoMuscoTsourakakis:2018}.
	\item \textbf{Facebook Egograph }\cite{LeskovecJureMcauley:2012}\\
	\textit{$n = 3964$ nodes, $e = 88159$ edges}\\
	Consists of ten anonymized ego networks, which are social circles of Facebook users -- the ten overlapping networks are combined into a single connected component.
\end{itemize}

\section{Additional Plots}\label{app:plots}

\begin{figure}[H]
	\minipage{0.49\textwidth}
	\includegraphics[width=\linewidth]{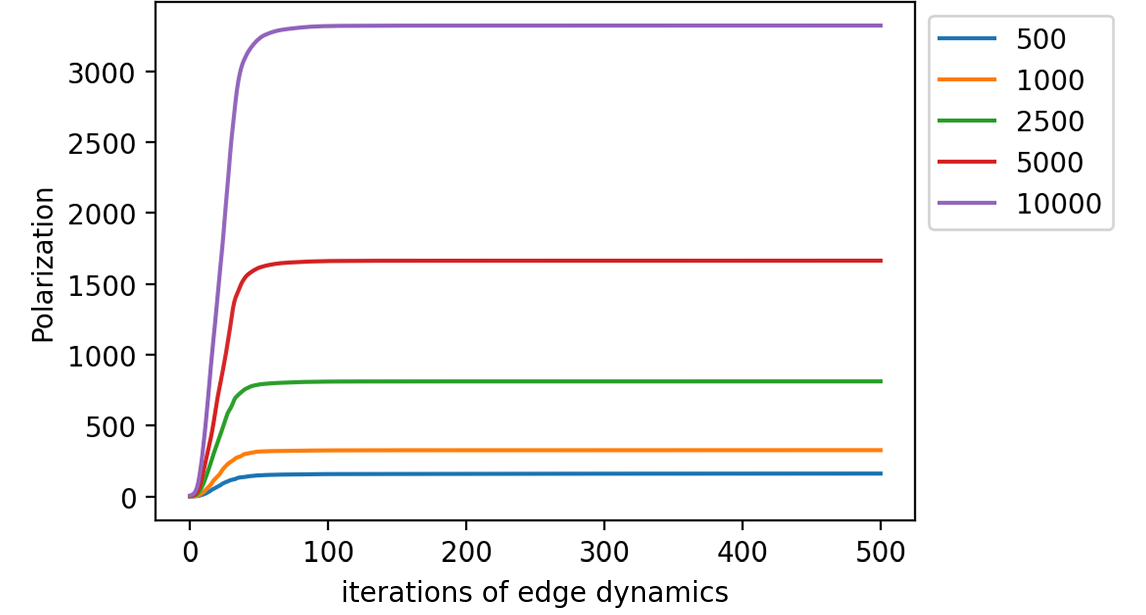}
	\endminipage
	\vspace{-.5em}
	\caption{Polarization over time for varying sizes $n$ of random ER graphs with average degree $25$, yielding $25(n-1)/2$ expected edges, and no fixed edges.  Dynamics include both friend-of-friend recommendations and confirmation bias. Innate opinions are distributed uniformly on the interval $[-1,1]$.  While the asymptotic value of polarization changes proportionally to the number of nodes in the graph, the polarizing behavior of edge dynamics generalizes across different graph sizes.}\label{fig:sizes}
	\minipage{0.40\textwidth}
	\includegraphics[width=\linewidth]{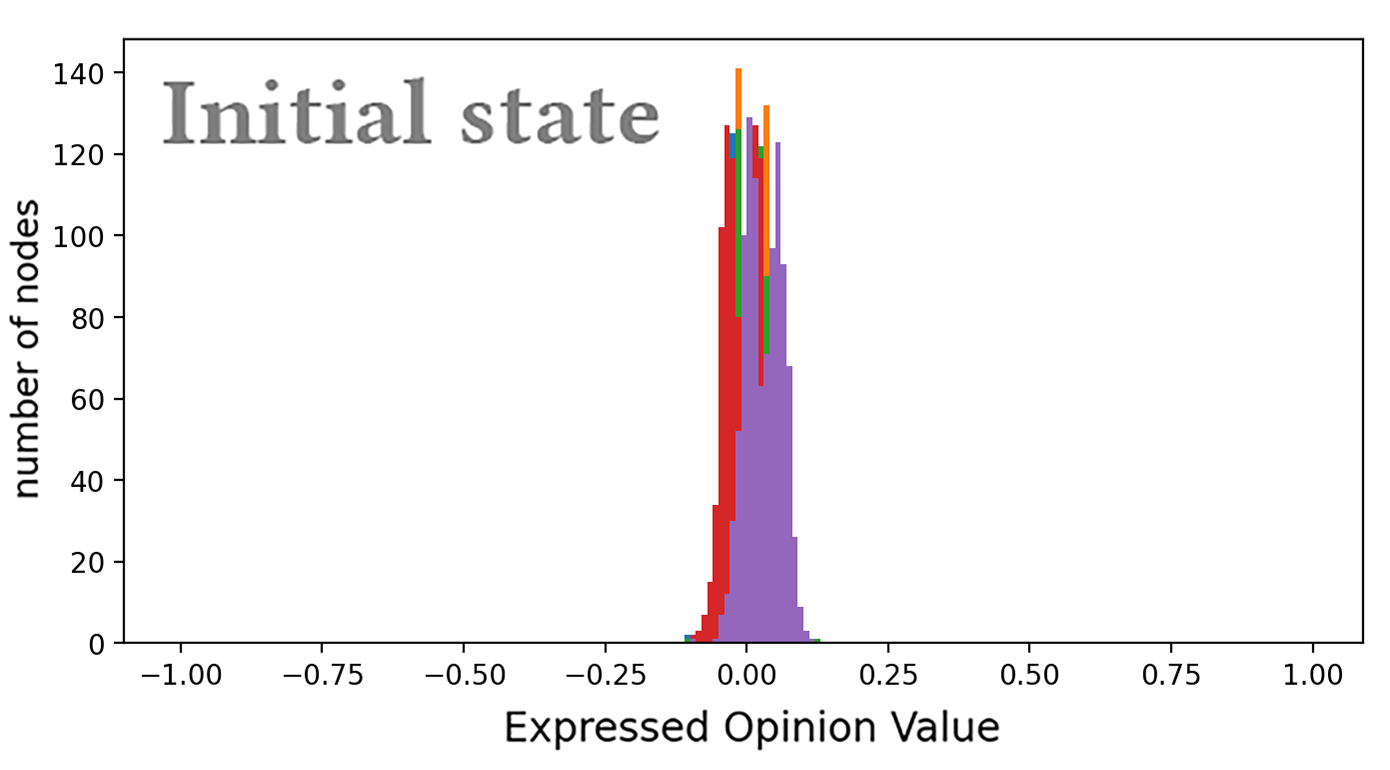}
	\endminipage\hfill
	\minipage{0.40\textwidth}
	\includegraphics[width=\linewidth]{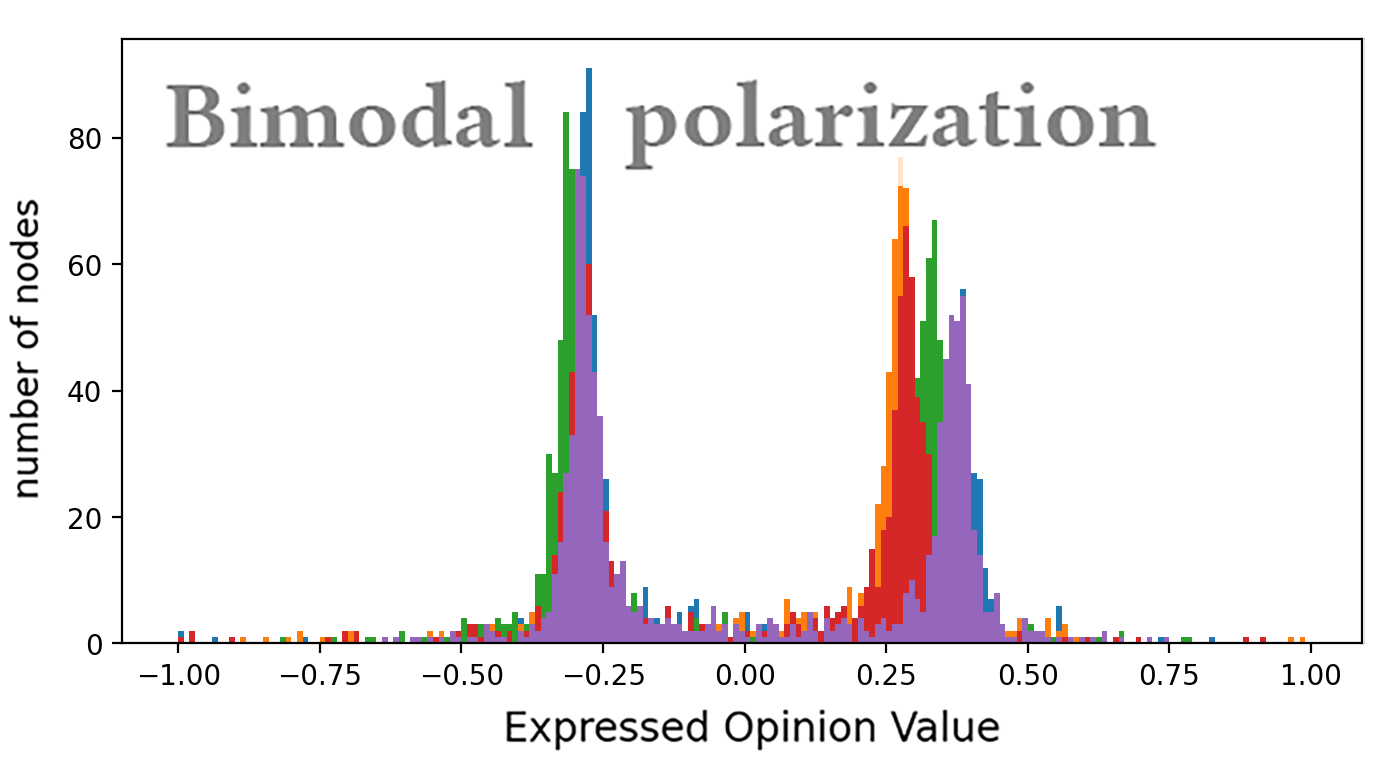}
	\endminipage\hfill
	\caption{Two states of opinion polarization for an Erd\"{o}s-Renyi graph with $n = 1000$ nodes, uniform innate opinions on the interval $[-1,1]$, and 5\% of edges fixed. Plots are histograms of the expressed opinions at time steps $t=5$ (top) and $t=120$ (bottom). Each colors represent one of 5 independent trials. In the \textit{initial state}, the opinions center around $0$, due to opinion averaging. The second plot shows the \textit{bimodal polarization} state, where opinions cluster into groups on either side of the mean.  Since the presence of fixed edges prevents the network from splintering further, this bimodal state is the steady state in the network.}\label{fig:hist1}
	\vspace{-.5em}
\end{figure}

\begin{figure}[H]
	\minipage{0.45\textwidth}
	\includegraphics[width=\linewidth]{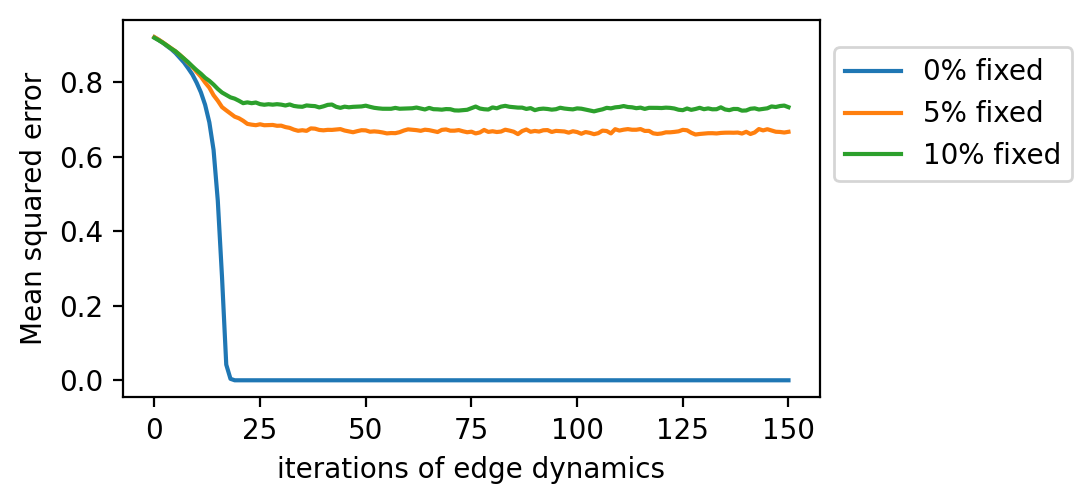}
	\endminipage
	\vspace{-.5em}
	\caption{Mean squared error over time between expressed opinions and innate opinions for an Erd\"{o}s-Renyi graph, with varying percentages of \emph{fixed edges}, as described in Fig. \ref{fig:fixedIntro}.  Innate opinions are uniformly assigned from $\{-1,1\}$.  When fixed edges are not present, the network can splinter and sort nodes such that expressed opinions are indistinguishable from innate opinions.  Otherwise, the mean squared error between expressed and innate opinions is reduced proportionally to the rise in polarization.}\label{fig:mse}
	\vspace{1em}
	\minipage{0.45\textwidth}
	\includegraphics[width=\linewidth]{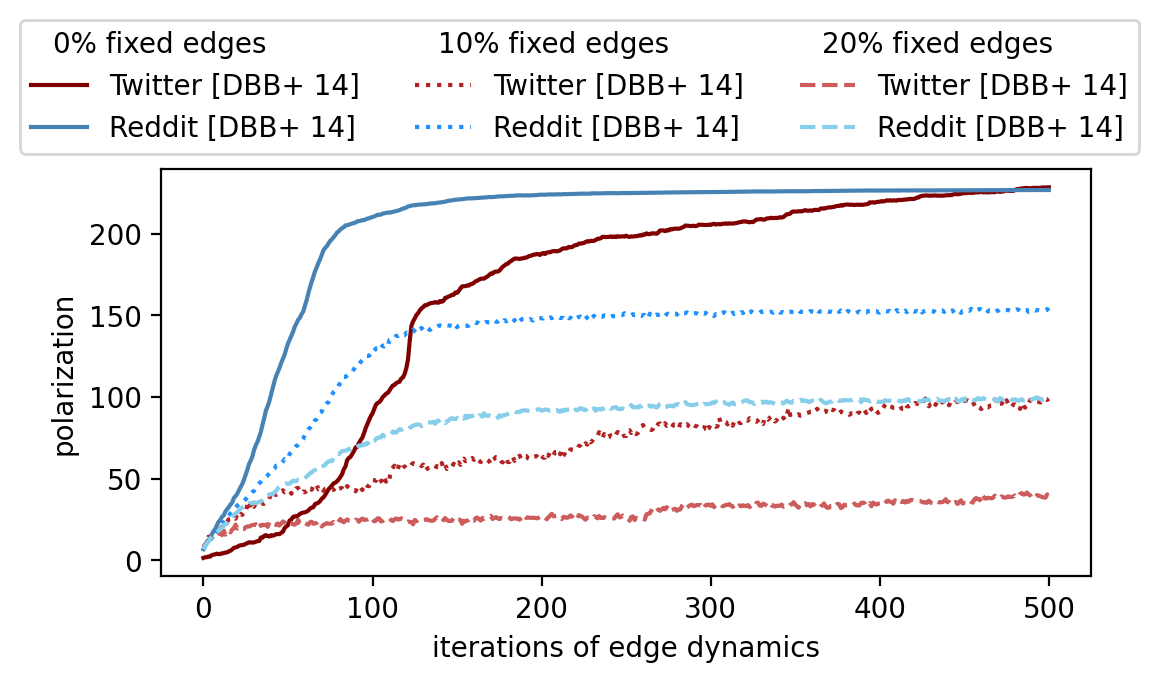}
	\endminipage
	\vspace{-1em}
	\caption{Varying percentages of fixed edges, on Twitter and Reddit real-world networks \cite{DeAbir:2014}. Edge dynamics include both friend-of-friend recommendations and confirmation bias.}\label{fig:fixed2}
	\vspace{1em}
	\minipage{0.40\textwidth}
	\includegraphics[width=\linewidth]{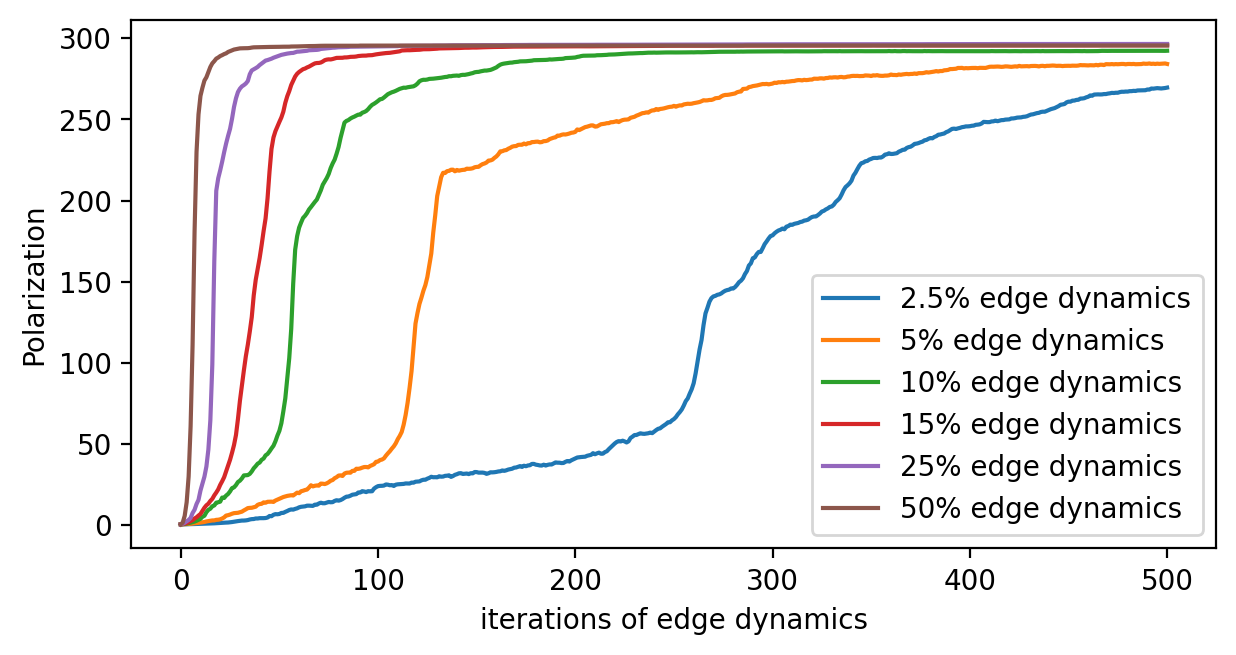}
	\endminipage
	\caption{Polarization over time for an ER graph with $n=1000$ nodes and fixed degree $25$, yielding $25(n-1)/2$ expected edges. Dynamics include both friend-of-friend recommendations and confirmation bias. We vary the percentage of edges which are added and removed at each time step, showing that the choice of this percentage simply scales the amount of time steps necessary for convergence.   Innate opinions are distributed uniformly on the interval $[-1,1]$.  }\label{fig:conn}
	\vspace{-.5em}
\end{figure}

\begin{figure*}
	\minipage{0.32\textwidth}
	\includegraphics[width=\linewidth]{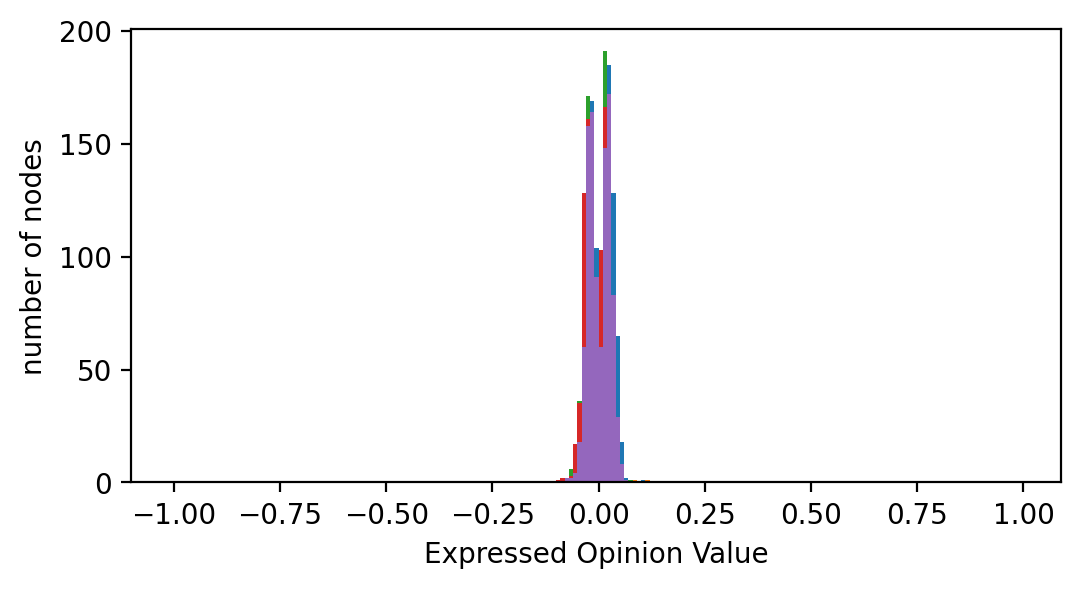}
	\vspace{-2em}
	\caption*{Initial state}
	\endminipage\hfill
	\minipage{0.32\textwidth}
	\includegraphics[width=\linewidth]{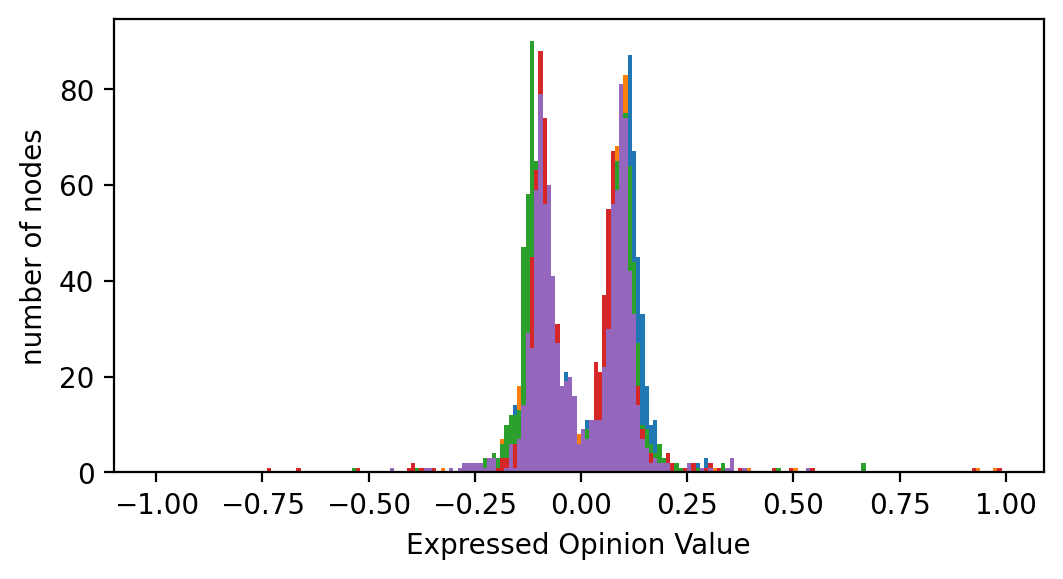}
	\vspace{-2em}
	\caption*{Bimodal polarization}
	\endminipage\hfill
	\minipage{0.32\textwidth}%
	\includegraphics[width=\linewidth]{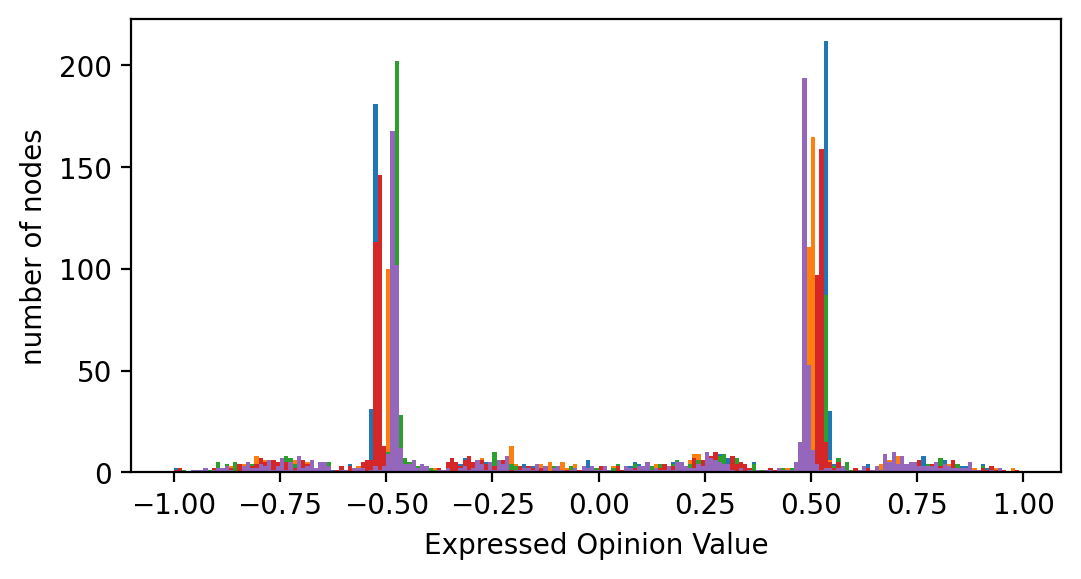}
	\vspace{-2em}
	\caption*{Maximal polarization}
	\endminipage\hfill
	\vspace{-0.75em}
	\caption{States of opinion polarization for an Erd\"{o}s-Renyi graph with no fixed edges, innate opinions sampled from mixture of two Gaussians with $\mu = \pm0.5$ and $\sigma = 0.2$, on interval $[-1,1]$. Plots are histograms of the expressed opinions at $t=1$, $t=20$, and $t=300$ time steps, from left to right. Colors represent 5 independent trials.}\label{fig:histBim}
\end{figure*}

\begin{figure}[H]
	\minipage{0.45\textwidth}
	\includegraphics[width=\linewidth]{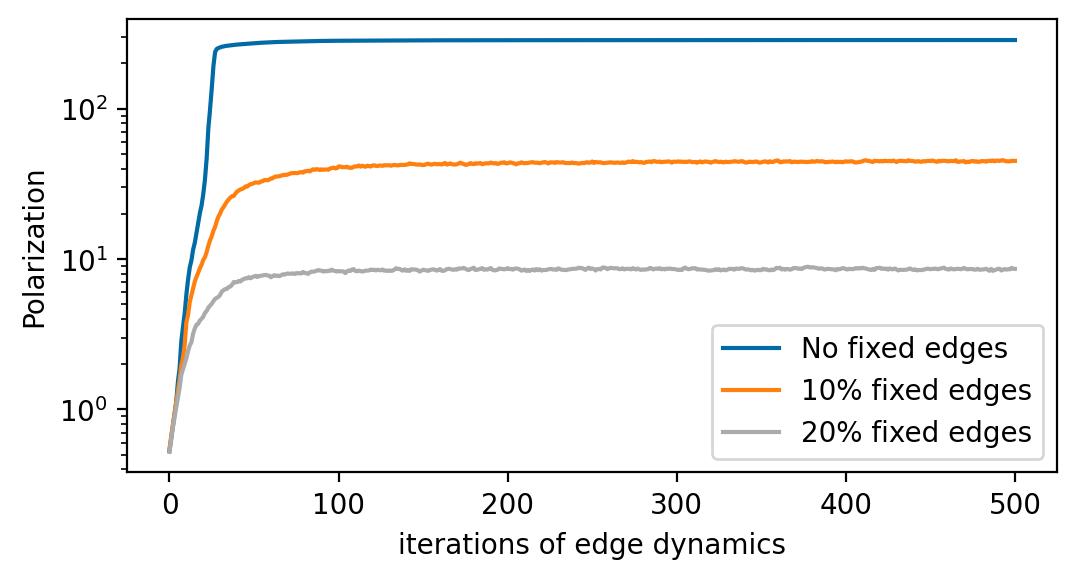}
	\endminipage
	\vspace{-1em}
	\caption{Polarization over time for an an Erd\"{o}s-Renyi graph with $1000$ nodes, varying percentages of \emph{fixed edges}, and average degree $25$, yielding $25(n-1)/2$ expected edges. Dynamics include both friend-of-friend recommendations and confirmation bias. Innate opinions are sampled from mixture of two Gaussians with $\mu = \pm0.5$ and $\sigma = 0.2$, on interval $[-1,1]$.  In this different setting for the innate opinions, polarization reaches a steady state slightly faster, but the behavior of our model is similar.}\label{fig:bimodalPL}
	\vspace{1em}
	\minipage{0.45\textwidth}
	\includegraphics[width=\linewidth]{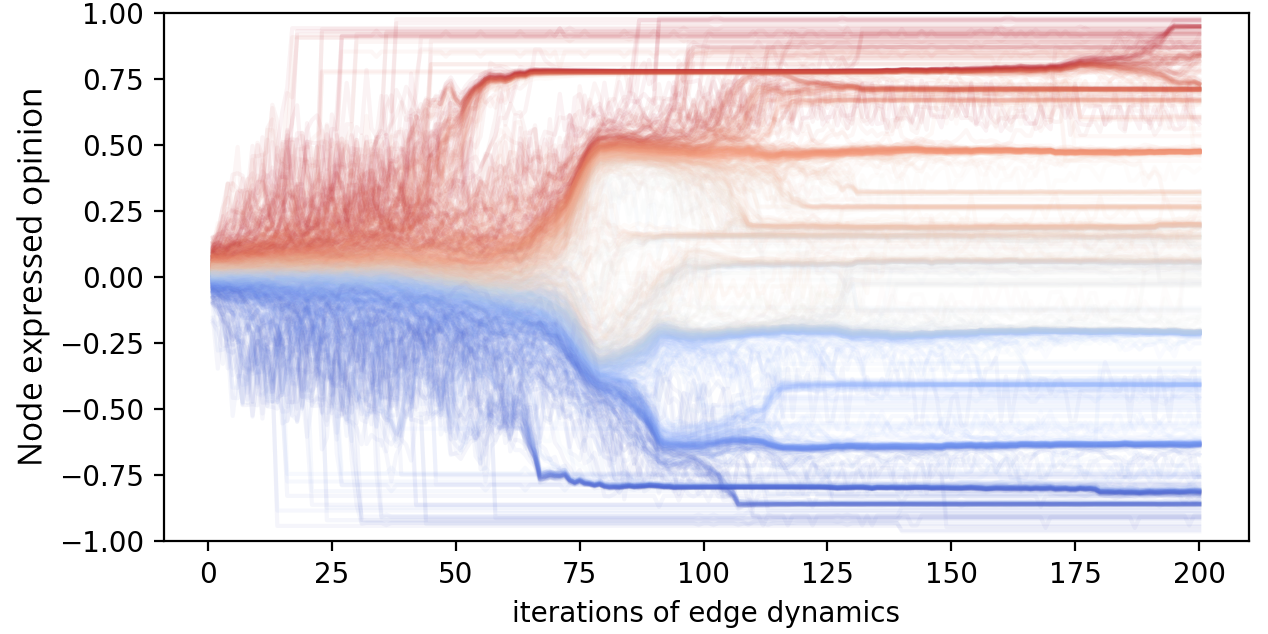}
	\endminipage
	\vspace{-.5em}
	\caption{Evolution of individual opinions on a Barab\'{a}si-Albert graph with $1000$ nodes, $m = 10$, and no fixed edges, with uniform random innate opinions.  Each line represents the evolving expressed opinion of one node. The color of each  line represents that node's innate opinion, ranging on a gradient from -1 (blue) to 1 (red). We see opinions initially grouped near $0$ (\textbf{Initial state}).  We then see a bifurcation of the major opinion clusters, denoted by the darker \& more saturated lines (\textbf{Bimodal polarization}), and finally, a spreading of opinions into clusters based on their innate opinions, representing  (\textbf{Maximal polarization}).}\label{fig:branch2}
\end{figure}

\begin{figure}[H]
	\minipage{0.4\textwidth}
	\vspace{-1em}
	\includegraphics[width=\linewidth]{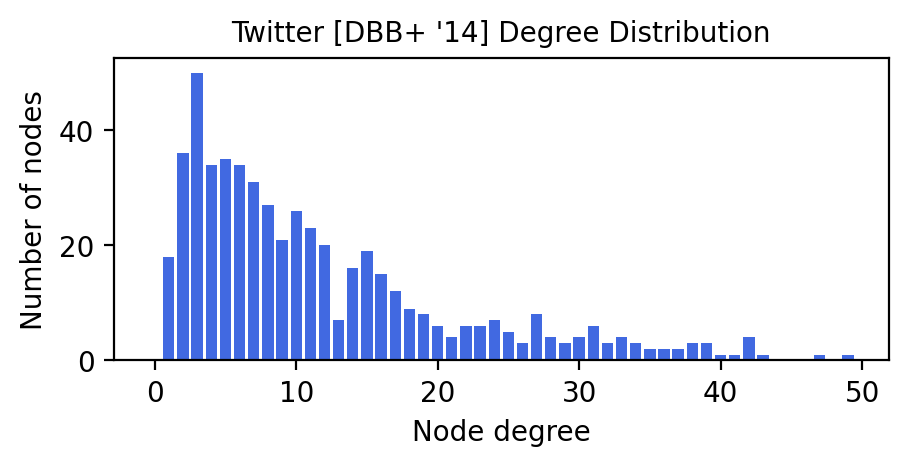}
	\vspace{-2.5em}
	\endminipage
	\caption{Degree distribution histogram for Twitter data set \cite{DeAbir:2014}.  The $x$-axis denotes the degree, and the height of each bar is the number of nodes with that degree.}\label{fig:twitDeg}
	\vspace{0em}
	\minipage{0.4\textwidth}
	\vspace{2em}
	\includegraphics[width=\linewidth]{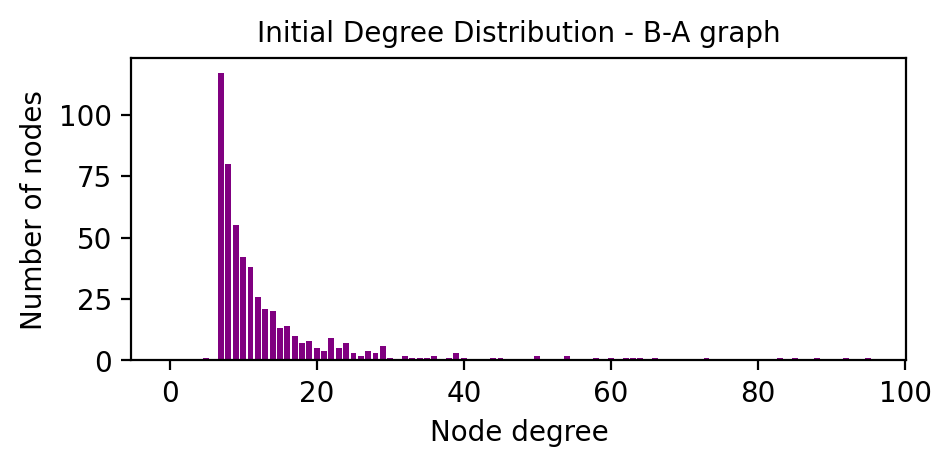}
	\endminipage\\
	\minipage{0.4\textwidth}
	\includegraphics[width=\linewidth]{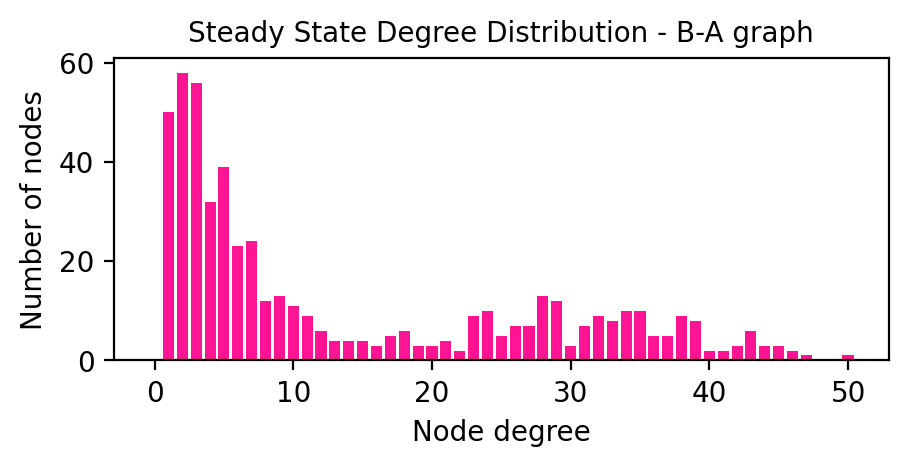}
	\vspace{-1em}
	\endminipage
	\caption{Initial and steady state degree distribution histograms for BA generated network, validating against Twitter data set \cite{DeAbir:2014}.  The $x$-axis denotes the degree, and the height of each bar is the number of nodes with that degree.  The steady state histogram preserves a power law degree distribution, reflecting what can be considered a realistic social network structure.}\label{fig:BAvalidDeg}
	\vspace{-1em}
\end{figure}

\begin{figure}[H]
	\minipage{0.4\textwidth}
	\vspace{-1em}
	\includegraphics[width=\linewidth]{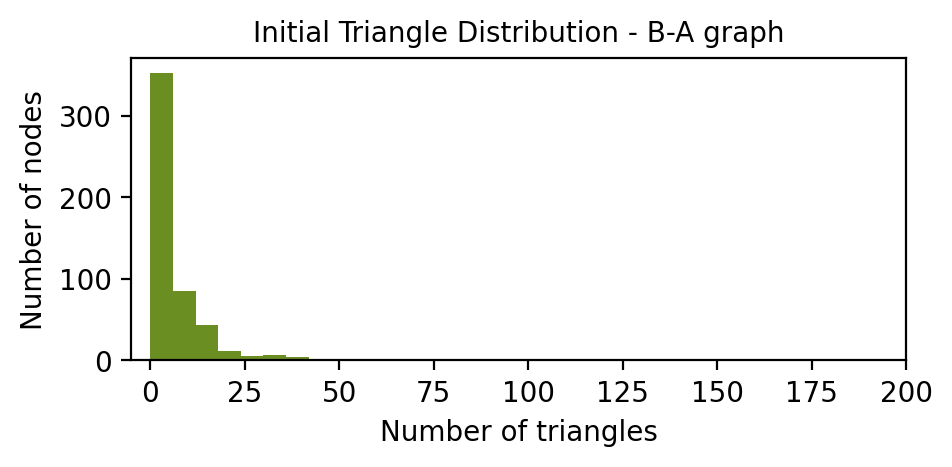}
	\endminipage\\
	\minipage{0.4\textwidth}
	\includegraphics[width=\linewidth]{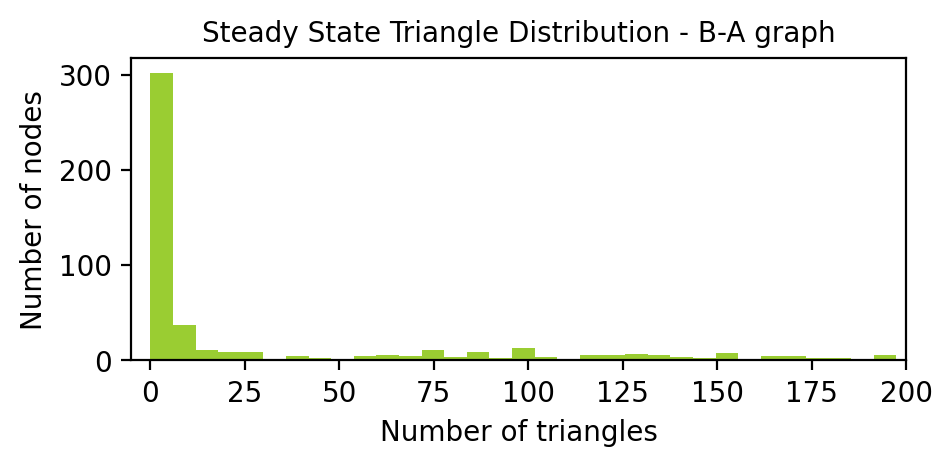}
	\vspace{-1em}
	\endminipage
	\caption{Initial and steady state triangle distribution histograms for BA generated network, validating against Twitter data set \cite{DeAbir:2014}.  The $x$-axis denotes the number of triangles, and the height of each bar is the number of nodes with that triangle count.  The steady state histogram shows that edge dynamics preserve the distribution of triangles in the initial graph, which is similar to the distribution in the Twitter data set.}\label{fig:BAvalidTri}
	\vspace{1em}
	\minipage{0.4\textwidth}
	\includegraphics[width=\linewidth]{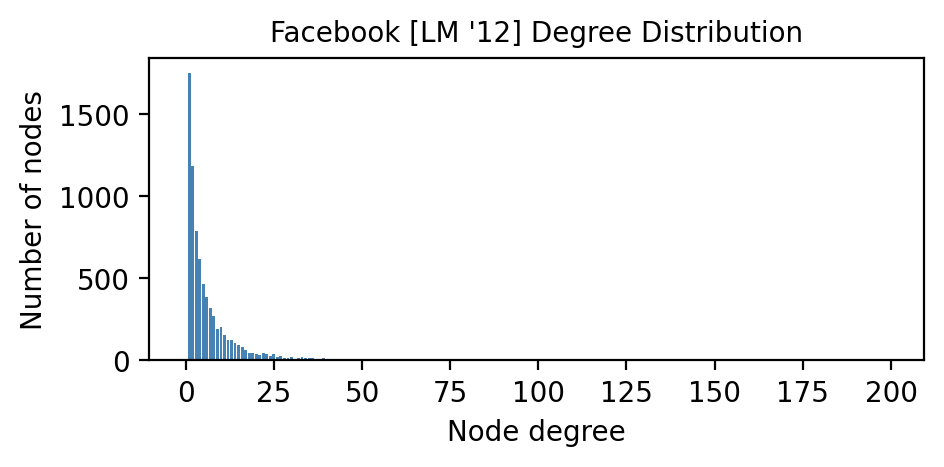}
	\vspace{-2.5em}
	\endminipage
	\caption{Degree distribution histogram for Facebook data set \cite{LeskovecJureMcauley:2012}.  The $x$-axis denotes the degree, and the height of each bar is the number of nodes with that degree.}
	\vspace{1em}
	\minipage{0.45\textwidth}%
	\includegraphics[width=\linewidth]{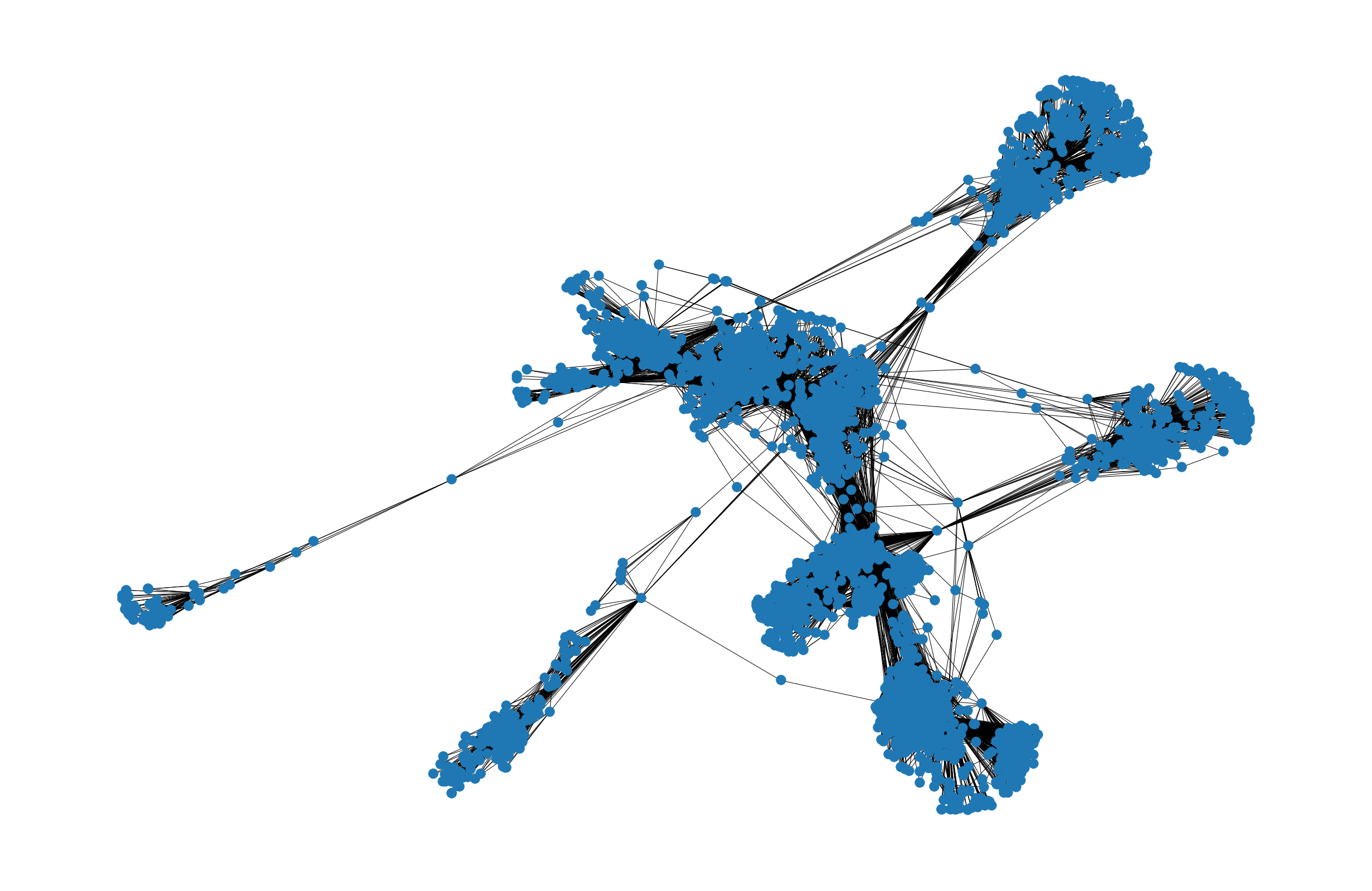}
	\caption{Facebook Egograph \cite{LeskovecJureMcauley:2012} real-world snapshot.}\label{fig:facebook}
	\endminipage\hfill
\end{figure}

\begin{figure}[H]
	\minipage{0.45\textwidth}
	\includegraphics[width=\linewidth]{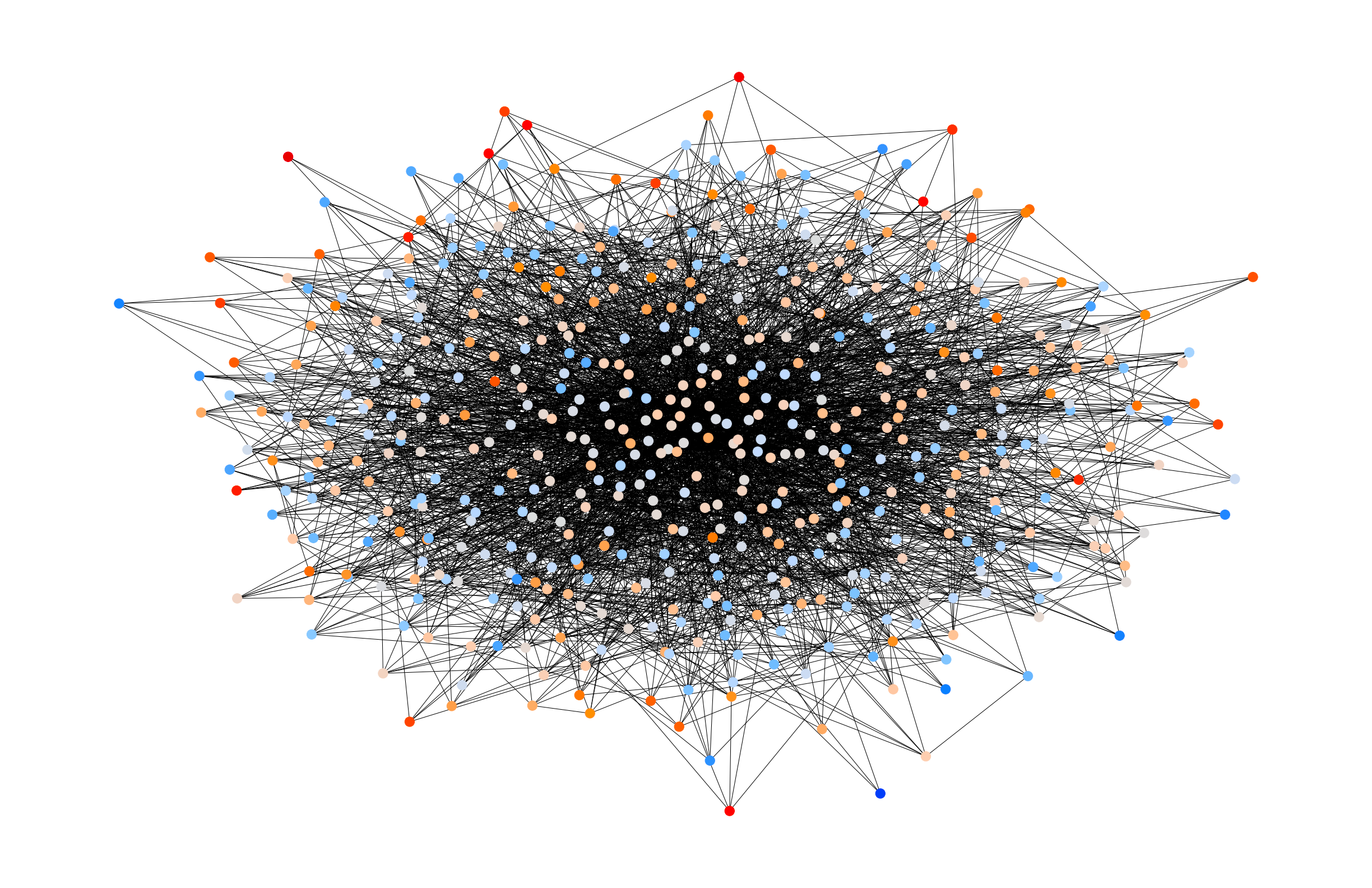}
	\vspace{-3em}
	\caption{\textit{Initial} Barab\'{a}si-Albert graph with 25\% fixed edges, before edge dynamics simulation.}\label{fig:baInit}
	\endminipage\newline
	\minipage{0.45\textwidth}%
	\includegraphics[width=\linewidth]{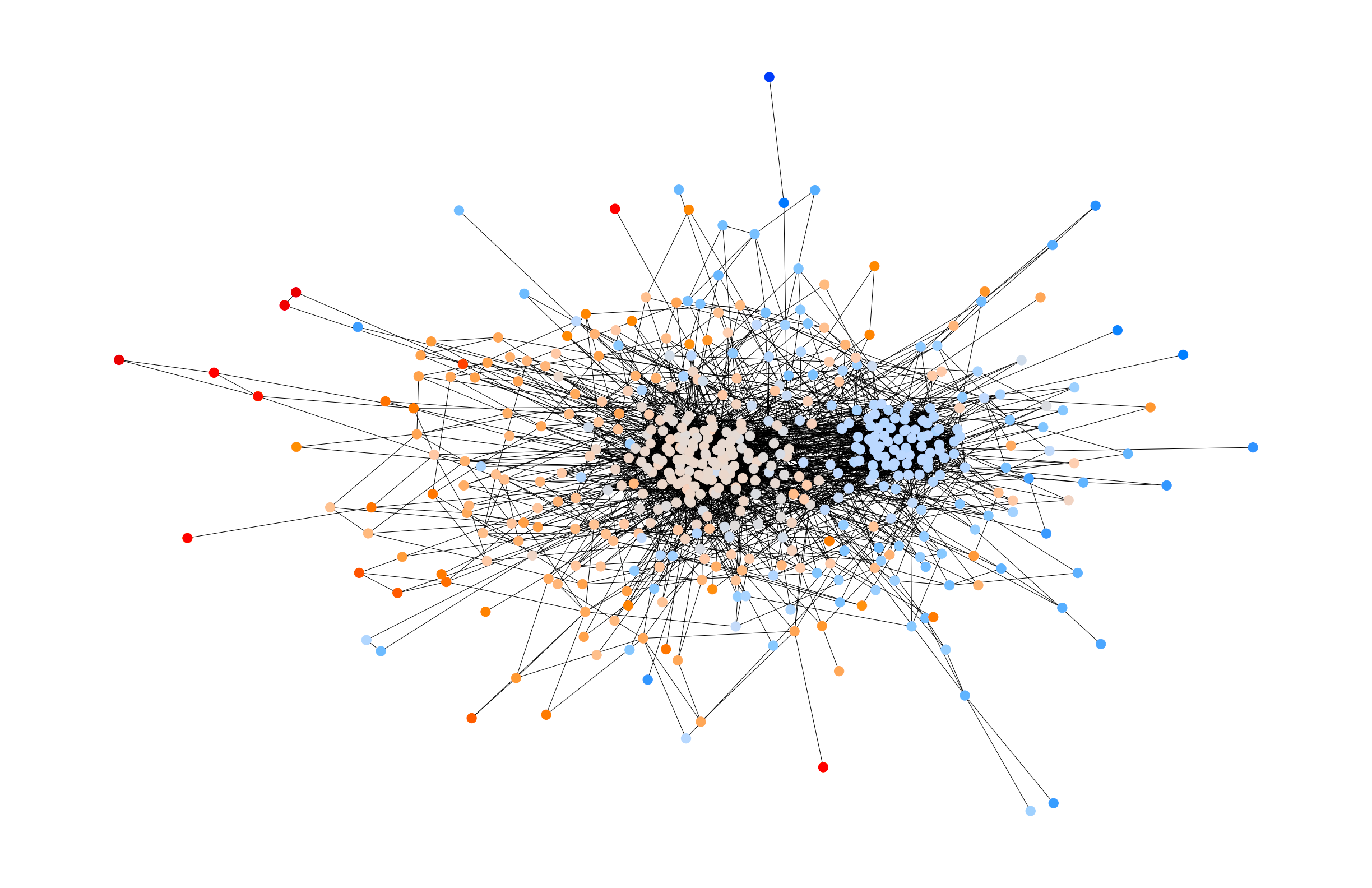}
	\vspace{-3em}
	\caption{\textit{Steady state} Barab\'{a}si-Albert graph with 25\% fixed edges, after edge dynamics simulation.}\label{fig:baConverged}
	\endminipage\newline
	\minipage{0.41\textwidth}
	\vspace{2em}
	\includegraphics[width=\linewidth]{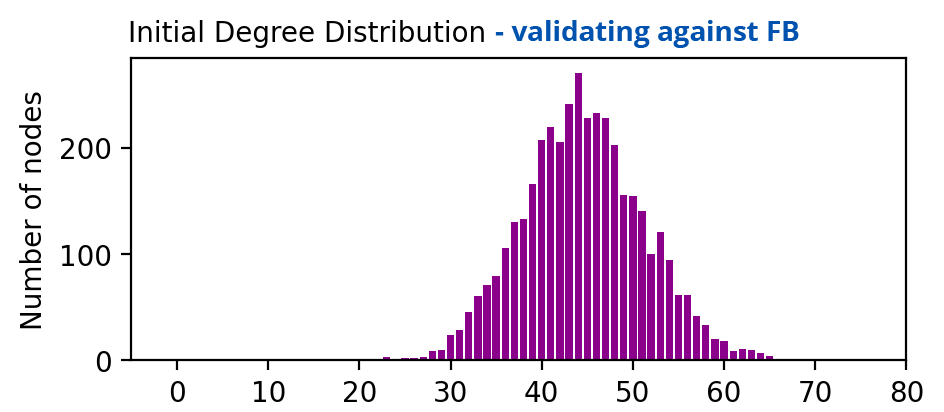}
	\endminipage\\
	\minipage{0.4\textwidth}
	\includegraphics[width=\linewidth]{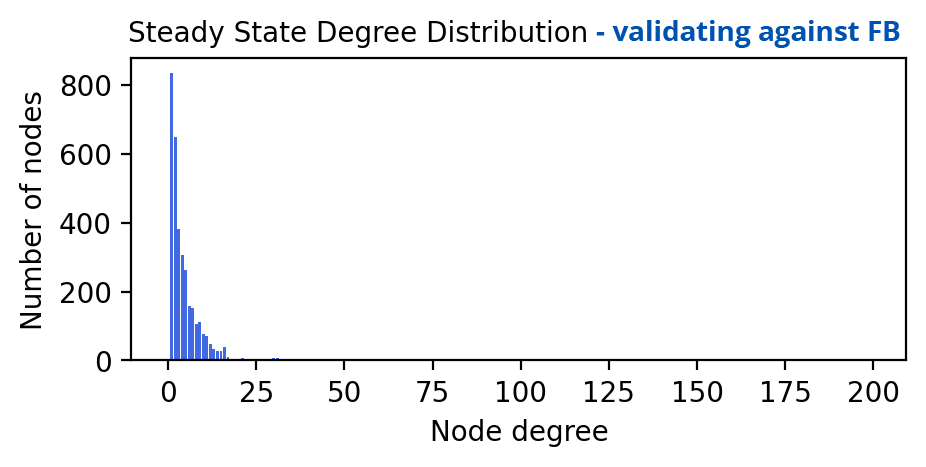}
	\vspace{-1.5em}
	\endminipage
	\caption{Initial and steady state degree distribution histograms for ER generated network, validating against Facebook data set \cite{LeskovecJureMcauley:2012}.  The $x$-axis denotes the degree, and the height of each bar is the number of nodes with that degree.  The steady state distribution differs significantly from the initial distribution. It is closer to a power law distribution, reflecting a more realistic network structure.}\label{fig:FacevalidDeg}
	\vspace{-1em}
\end{figure}

\begin{figure}[H]
	\minipage{0.41\textwidth}
	\includegraphics[width=\linewidth]{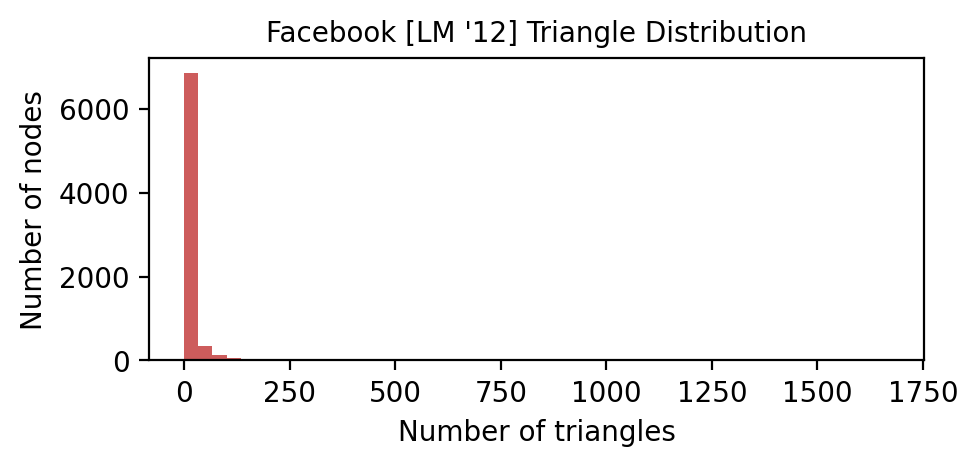}
	\vspace{-2.5em}
	\endminipage
	\caption{Triangle distribution histogram for Facebook data set \cite{DeAbir:2014}.  The $x$-axis denotes the number of triangles incident on a node, and the height of each bar is the number of nodes with that number of triangles.}
	\vspace{2em}
	\minipage{0.41\textwidth}
	\vspace{-1em}
	\includegraphics[width=\linewidth]{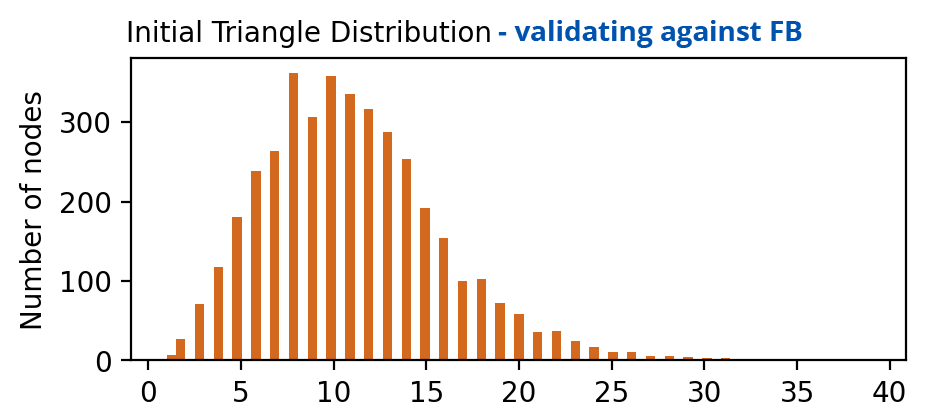}
	\endminipage\\
	\minipage{0.43\textwidth}
	\includegraphics[width=\linewidth]{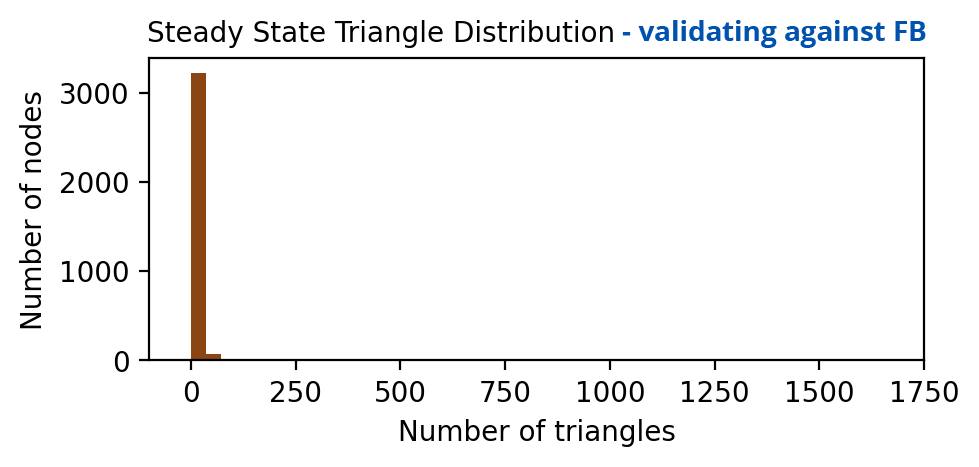}
	\vspace{-1em}
	\endminipage
	\caption{Initial and steady state triangle distribution histograms for BA generated network, validating against Twitter data set \cite{DeAbir:2014}.  The $x$-axis denotes the number of triangles, and the height of each bar is the number of nodes with that triangle count.  The steady state histogram shows that edge dynamics preserve the distribution of triangles in the initial graph, which is similar to the distribution in the Twitter data set.}\label{fig:FacevalidTri}
\end{figure}

\begin{figure}[H]
	\minipage{0.5\textwidth}
	\includegraphics[width=\linewidth]{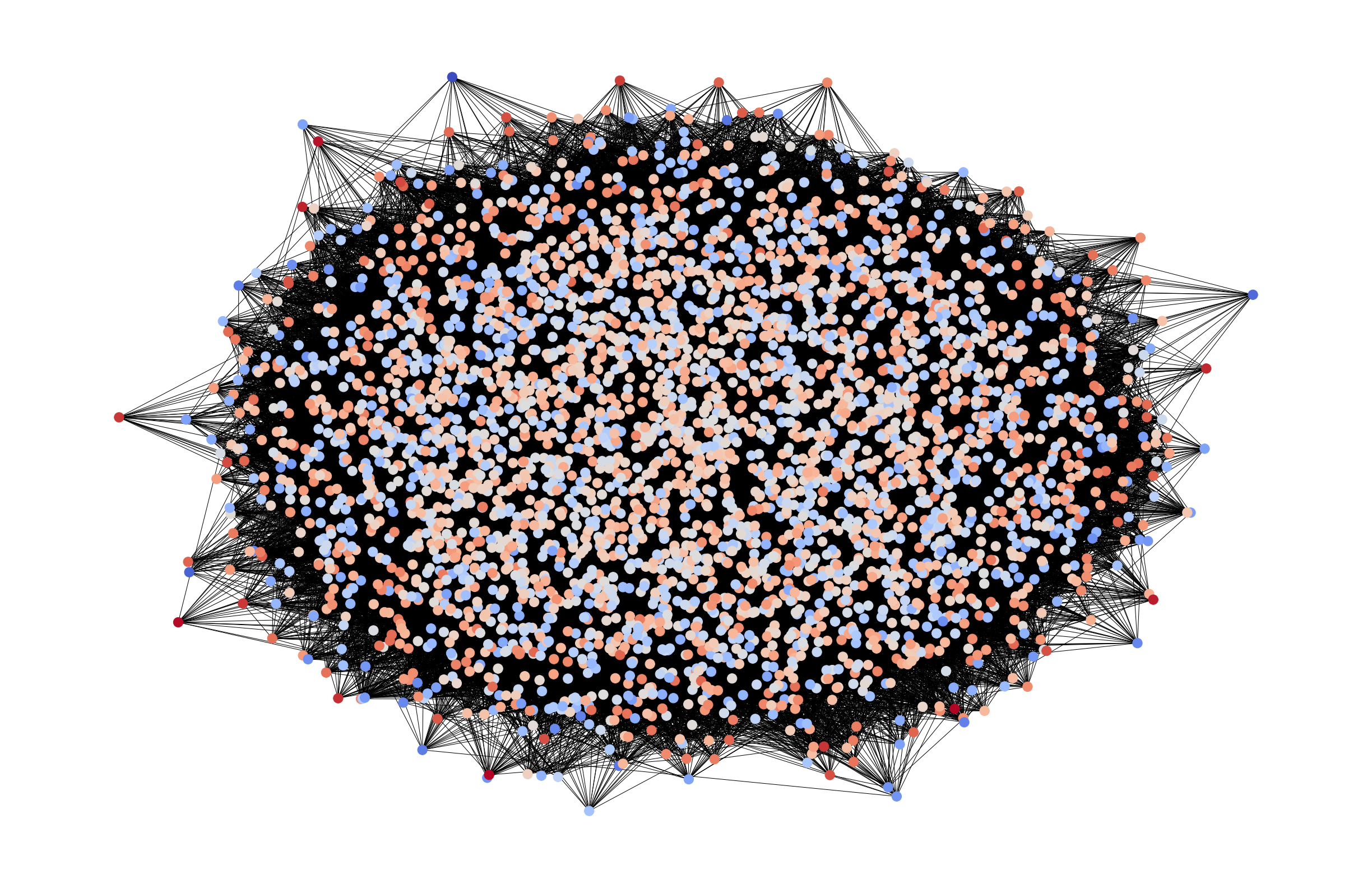}
	\caption{\textit{Initial} Erd\"{o}s-Renyi graph with $3964$ nodes, $87872$ edges, and 5\% fixed edges, validating against Facebook data set; visualized before edge dynamics simulation.}\label{fig:erFaceInit}
	\endminipage\newline
	\minipage{0.5\textwidth}%
	\includegraphics[width=\linewidth]{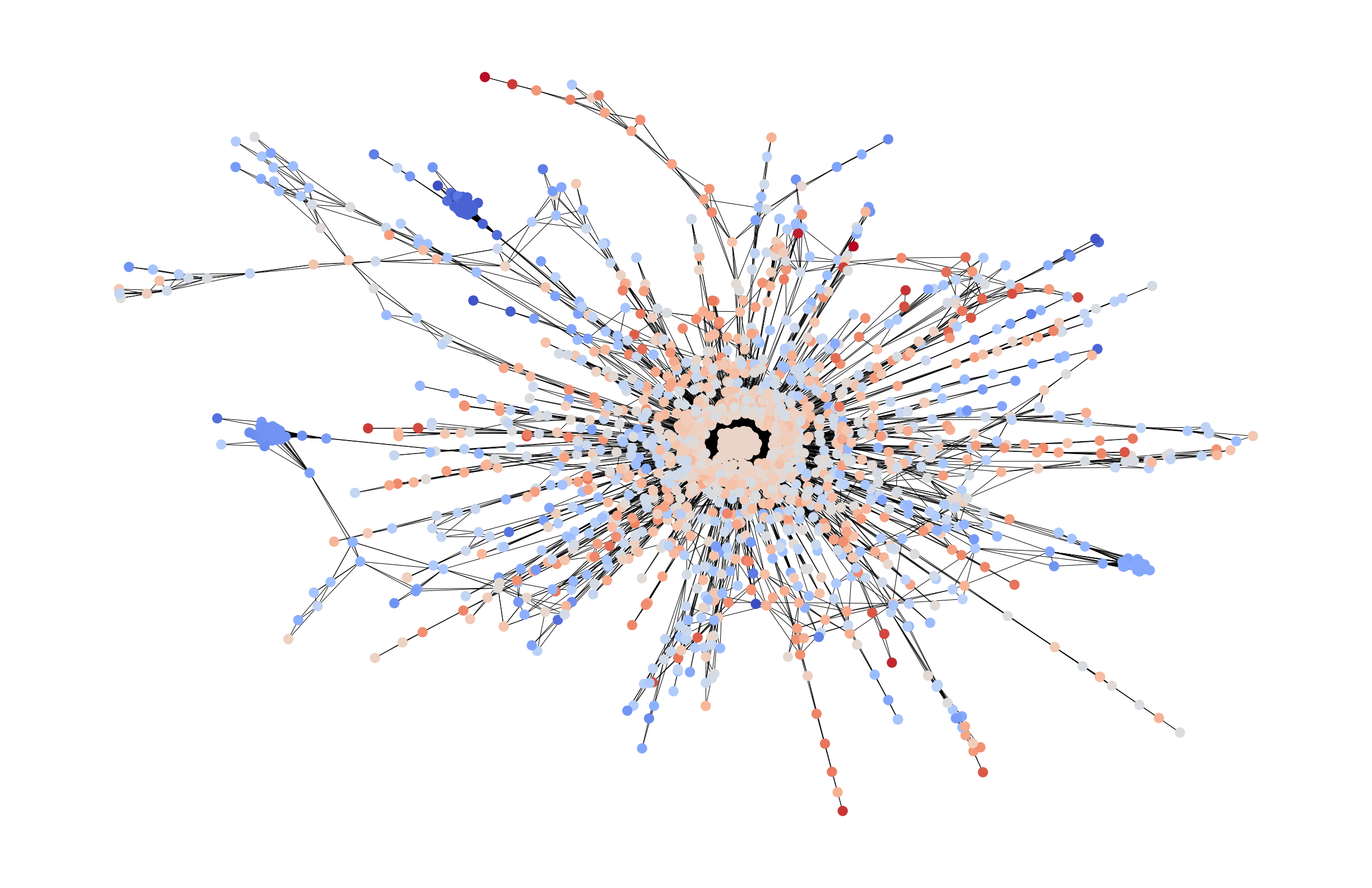}
	\caption{\textit{Steady state} Erd\"{o}s-Renyi graph with $3964$ nodes, $87872$ edges, and 5\% fixed edges, validating against Facebook data set; visualized after edge dynamics simulation.}\label{fig:erFaceConverged}
	\endminipage
\end{figure}

\end{document}